\documentclass[journal, onecolumn]{IEEEtran}

\usepackage{url}
\usepackage{ifthen}
\usepackage{cite}
\usepackage[cmex10]{amsmath} 
\usepackage{bbold}
\usepackage{bookmark}

\usepackage{balance}
\usepackage{amsmath}
\usepackage{amsthm}
\usepackage{multirow}
\usepackage{amsfonts}
\usepackage{graphicx}
\usepackage{algorithm}
\usepackage{url}
\usepackage{lipsum}
\usepackage{booktabs,makecell,multirow}
\usepackage{tikz}
\usetikzlibrary{matrix}
\usepackage{rotating}
\usepackage{mathtools}
\usepackage{tkz-tab}

\usepackage{latexsym}
\usepackage{color}
\usepackage{graphicx}
\usepackage{threeparttable}
\usepackage{float}

\usepackage{algpseudocode}

\errorcontextlines\maxdimen

\newcounter{algsubstate}
\renewcommand{\thealgsubstate}{\alph{algsubstate}}
\newenvironment{algsubstates}
  {\setcounter{algsubstate}{0}%
   \renewcommand{\State}{%
     \stepcounter{algsubstate}%
     \Statex {\footnotesize\thealgsubstate:}\space}}
  {}

\DeclareGraphicsExtensions{.tif,.pdf,.jpeg,.png,.jpg,.bmp,.svg,.eps,.EMF}

\newtheorem{theorem}{Theorem}
\newtheorem{lemma}{Lemma}
\newtheorem{problem}{Problem}

\newcommand{\cE}{\mathcal{E}}

\newcommand{\cM}{\mathcal{M}}
\newcommand{\cA}{\mathcal{A}}

\newcommand{\cB}{\mathcal{B}}

\newcommand{\bX}{\mathbf{x}}
\newcommand{\bY}{\mathbf{y}}

\newcommand{\bR}{\mathbf{r}}

\newcommand{\bZ}{\mathbf{z}}

\newcommand{\bF}{\mathbf{f}}

\newcommand{\poly}{\mathsf{poly}}

\newcommand{\supp}{\mathsf{supp}}

\newcommand{\test}{\mathsf{test}}

\DeclarePairedDelimiter\ceil{\lceil}{\rceil}
\DeclarePairedDelimiter\floor{\lfloor}{\rfloor}

\begin{document}

\title{Concomitant Group Testing}

\author{Thach V. Bui~\IEEEmembership{Member, IEEE}, Jonathan Scarlett~\IEEEmembership{Member, IEEE}
\thanks{Thach V. Bui was with the Department of Computer Science, National University of Singapore, Singapore. He is currently with the Faculty of Engineering, National University of Singapore, Singapore. Email: bvthach@nus.edu.sg.}
\thanks{Jonathan Scarlett is with the Department of Computer Science, Department of Mathematics, and Institute of Data Science, National University of Singapore, Singapore. Email: scarlett@comp.nus.edu.sg.}
}

\maketitle

\thispagestyle{plain}
\pagestyle{plain}


\IEEEpeerreviewmaketitle

\begin{abstract}    
    In this paper, we introduce a variation of the group testing problem capturing the idea that a positive test requires a combination of multiple ``types'' of item.  Specifically, we assume that there are multiple disjoint \emph{semi-defective sets}, and a test is positive if and only if it contains at least one item from each of these sets.  The goal is to reliably identify all of the semi-defective sets using as few tests as possible, and we refer to this problem as \textit{Concomitant Group Testing} (ConcGT).  We derive a variety of algorithms for this task, focusing primarily on the case that there are two semi-defective sets.  Our algorithms are distinguished by (i) whether they are deterministic (zero-error) or randomized (small-error), and (ii) whether they are non-adaptive, fully adaptive, or have limited adaptivity (e.g., 2 or 3 stages).  Both our deterministic adaptive algorithm and our randomized algorithms (non-adaptive or limited adaptivity) are order-optimal in broad scaling regimes of interest, and improve significantly over baseline results that are based on solving a more general problem as an intermediate step (e.g., hypergraph learning).
\end{abstract}

\section{Introduction}
\label{sec:intro}

The main goal of group testing is to efficiently identify an unknown subset of items in a large population of items. A test on a subset of items is positive or negative depending on the testing model and the items included in the test. The procedure for choosing placements of items into tests can be viewed as an \textit{encoding} procedure, and classifying the items based on the test outcomes is the \textit{decoding} procedure. It is desirable to minimize the number of tests, the number of stages of adaptivity, and the decoding time.

In standard group testing, the subset of interest is contains all defective items, and the other items are referred to as non-defective items. The outcome of a test on a subset of items is positive if the subset has at least one defective item and negative otherwise.  A number of related models also exist; for example, in threshold group testing with a single threshold, the outcome of a test on a subset of items is positive if the number of defective items in the subset is equal or larger than the threshold and negative otherwise. 

One of the most general variants of group testing is \emph{complex group testing}. In this model, there are subsets of items $D_1, \ldots, D_c$ such that a test on a set of items is positive if at least one of these $D_i$ is a subset of the tested set, and negative otherwise.  The items in $D = D_1 \cup \ldots \cup D_c$ are often referred to as semi-defective items. Complex GT originated in molecular biology~\cite{torney1999sets}, and follow-up works include~\cite{chang2010identification,chen2008upper,chin2013non}. Angluin~\cite{angluin2008learning}, Abasi et al.~\cite{abasi2018non}, and Abasi and Nader~\cite{abasi2019learning} also considered an equivalent problem posed as learning a hidden hypergraph, where the preceding $D_i$ are interpreted as hyperedges.

In this work, we consider a problem that can be viewed as a special case of Complex GT or learning a hidden hypergraph, but whose special structure leads to a considerably smaller number of tests compared to general methods.  Specifically, we consider the case that the sets $\{ D_1, \ldots, D_c \}$ are the elements of $S_1 \times \ldots \times S_m$ for $m$ disjoint subsets $S_1,\dotsc,S_m$, meaning that a test is positive if and only if it contains at least one item from each of these sets.  We proceed with some motivation, the formal problem formulation, connections to the more general settings described above, and our main contributions.

\subsection{Motivation}
\label{sub:intro:motiv}

At a high level, our problem is motivated by two closely related phenomena that may arise depending on the application:
\begin{itemize}
    \item A certain ``failure'' is observed if multiple distinct kinds of ``defects'' occur simultaneously;
    \item A certain ``target'' is attained if multiple distinct kinds of items/users are present to ``co-operate'' in attaining the target.
\end{itemize}

For consistency with the existing group testing literature, we adopt terminology corresponding to former of these.  Each group can be treated as a semi-defective group, and each member of the group is considered as a semi-defective item. A test on a subset of items is positive if the subset contains \textit{at least one item} in each group, and negative otherwise.  Returning to the above two dot points, a positive test result in the first dot point indicates ``an overall failure is observed'', whereas a positive result in the second dot point indicates ``the target is achieved''.  Because of the requirement of different types of items complementing each other, we call this problem \textit{Concomitant Group Testing} (ConcGT).

We now outline some potential applications as motivation for the ConcGT problem:
\begin{itemize}
    \item The edge detecting problem was motivated by learning which pairs (or triplets, etc.) of chemicals react with each other \cite{bouvel2005combinatorial,abasi2019learning}.  The ConcGT problem can be motivated similarly, with the added idea that the chemicals have associated categories, and including one chemical from each category is both necessary and sufficient to produce a reaction.
    \item Group testing has been proposed as a method for detecting attackers in security applications, such as jamming a wireless network \cite{shin2009reactive}.  ConcGT may similarly be relevant in scenarios where multiple kinds of attackers need to co-operate in order to break the system (e.g., by breaking multiple distinct defenses).
    \item In devising medical treatments, the efficacy of a treatment may require a combination of drugs each serving a different purpose, e.g., targeting specific cells or signaling pathways \cite{NCI,rao2011modern}.  A natural problem is then to devise an effective drug combination without knowing \emph{a priori} which drugs serve which purpose.  The ConcGT problem corresponds to a situation where we are unable to observe the individual drug effects separately, but instead can only observe whether the \emph{overall} combination was effective.
    \item In \cite{malioutov2017learning}, it was proposed to use group testing to learn simple and interpretable machine learning classification rules, e.g., rules expressed as a disjunction (OR operation) of single-feature classifiers.  ConcGT may benefit this application by allowing more flexible classification rules without losing too much interpretability, e.g., a conjunction (AND operation) of disjunctions of single-feature classifiers.
\end{itemize}
Naturally, most of these applications may require more sophisticated ConcGT-like models (e.g., to model noisy test results) instead of the simple noiseless model that we will consider.  However, we believe that this simpler model serves as the ideal starting point for introducing the problem and establishing some initial theoretical foundations.

\subsection{Problem formulation}
\label{sub:intro:prob}

We define the concomitant group testing problem as follows.

\begin{problem} (Concomitant group testing)
Given a population of $n$ items, $N = [n] = \{1,\dotsc,n\}$, there are $m \geq 2$ disjoint nonempty subsets $S_1, \ldots, S_m \subset N$. A test on a subset of items is positive if the subset contains at least one item in each $S_i$ for $i = 1, \ldots, m$, and negative otherwise. Our goal is to use the test results to identify $S_1, S_2, \ldots, S_m$ up to a permutation of the indices $1,\dotsc,m$, ideally with efficient decoding and as few tests as possible.\footnote{Re-ordering these indices leaves the test results unchanged, so in general there is no hope of guaranteeing the ``correct'' order.}
\label{prob:ConcGT}
\end{problem}

Note that if $m = 1$, ConcGT becomes standard group testing. We suppose that \textit{cardinality upper bounds of all semi-defective subsets $S_1, \ldots, S_m$ are known}, namely, $|S_i| \leq s_i$ for $i = 1, 2, \ldots, m$. When $m = 2$, we sometimes consider the case that they are exactly known (i.e., $|S_1| = s_1$ and $S_2 = s_2$) or known to within a constant factor (see Appendix \ref{app:KnownConst}).  Throughout the paper we will specify which results are for $m=2$ and which are for general $m$.

We refer to $S_1,\dotsc,S_m$ as \emph{semi-defective sets}, and we say that a size-$m$ collection of items is \emph{defective} if it contains one item from each of $S_1,\dotsc,S_m$, and \emph{non-defective} otherwise.  When $m=2$, we refer to these as \emph{defective pairs} and \emph{non-defective pairs}.

\textbf{Lower bound:} Here we derive an information-theoretic lower bound for ConcGT. Assume that $|S_i| = s_i$ for $i = 1, 2, \ldots, m$. 
Since all subsets are disjoint, there are $\binom{n - \sum_{v = 1}^{i - 1} s_v}{s_i}$ choices for subset $S_i$ for $i = 1, 2, \ldots, m$ after choosing $S_1, \ldots, S_{i - 1}$, where $s_0 = |S_0| = 0$. Therefore, there are $\prod_{i = 1}^m \binom{n - \sum_{v = 1}^{i - 1} s_v}{s_i}$ possibilities to choose $S_1, S_2, \ldots, S_m$. On the other hand, because each test has only two possible outcomes, the minimum number of tests required to uniquely identify all subsets $S_1, \ldots, S_m$ is:
\begin{equation}
{\log_2{\left( \prod_{i = 1}^m \binom{n - \sum_{v = 1}^{i - 1} s_v}{s_i} \right)}} = \Omega \left( { \sum_{i = 1}^m s_i \log{ \frac{n - \sum_{v = 1}^{i - 1} s_v}{s_i}}} \right).
\label{eqn:lower}
\end{equation}
Under the mild condition $\sum_{v = 1}^{m} s_v \leq \frac{n}{2}$ (or similarly with any constant in $(0,1)$ in place of $\frac{1}{2}$),~\eqref{eqn:lower} can be simplified to $\Omega \big( \sum_{i = 1}^m s_i \log{\frac{n}{s_i}} \big)$. Because $|S_i| = s_i$ for $i = 1, 2, \ldots, m$ is an instance of the case $|S_i| \leq s_i$ for $i = 1, 2, \ldots, m$, the result in~\eqref{eqn:lower} is also a lower bound for the more general case.  The following theorem formally states the resulting lower bound, and extends it to the case of non-zero error probability and a randomly varying number of tests.

\begin{theorem}
For the ConcGT problem with $m \ge 2$ disjoint nonempty subsets $S_1, \ldots, S_m \subset [n]$ with $|S_i| \leq s_i$, to deterministically recover $S_1, \ldots, S_m$ with a fixed number of tests (possibly adaptive), the required number of tests is at least
\begin{align}
\log_2{\left( \prod_{i = 1}^m \binom{n - \sum_{v = 1}^{i - 1} s_v}{s_i} \right)} &= \Omega \left( \sum_{i = 1}^m s_i \log{ \frac{n - \sum_{v = 1}^{i - 1} s_v}{s_i}} \right) \label{eq:lb} \\
&= \Omega \bigg( \sum_{i = 1}^m s_i \log{\frac{n}{s_i}} \bigg), \mbox{ if } \sum_{v = 1}^{m} s_v \leq \frac{n}{2}. 
\end{align}
Moreover, the same order-wise lower bound remains true even when the algorithm is given the following two kinds of flexibility: (i) it is randomized and may fail with some constant probability $\delta \in (0,1)$; and (ii) the number of tests is random and the lower bound is on the \emph{average} number of tests.
\label{thm:lower}
\end{theorem}
\begin{proof}
    For deterministic recovery with a fixed number of tests, the claim follows directly from the above counting argument.   For randomized designs with a fixed number of tests, we use Yao's minimax principle to consider uniformly random $S_1,\dotsc,S_m$ and a deterministic algorithm, upon which the proof becomes identical to the standard setting \cite[Thm.~3.1]{baldassini2013capacity} but with $n \choose k$ replaced by $\prod_{i = 1}^m \binom{n - \sum_{v = 1}^{i - 1} s_v}{s_i}$.  We omit the details, as this is now a very standard analysis.

    To prove the result even under the added flexibility of a variable number of tests, suppose for contradiction that such an algorithm existed with error probability at most $\delta \in (0,1)$ and an expected number of tests $\tau_{\rm avg}$ strictly smaller than the right-hand side of \eqref{eq:lb} (i.e., with $\Omega(\cdot)$ replaced by $o(\cdot)$).  Then consider an algorithm that runs the original algorithm, but stops and fails after $C\tau_{\rm avg}$ tests have been taken.  By Markov's inequality and the union bound, this algorithm has error probability at most $\delta + \frac{1}{C}$, which is still a constant in $(0,1)$ when $C$ is large enough.  Since the number of tests is at most $C\tau_{\rm avg}$ and $C$ is constant, this contradicts the order-wise lower bound for the case of a fixed number of tests.
\end{proof}

\subsection{Related work}
\label{sec:intro:related}

{\bf Summary of criteria.} There are many criteria considered when studying group testing. The first important criterion is the number of stages for deploying tests. When the number of stages is one, the design is non-adaptive and all tests are designed to test simultaneously. On the other hand, when there are many stages, the design is adaptive and a test depends on the designs of the previous tests. Adaptive designs often attain an information-theoretically optimal bound on the number of tests. Normally, it takes more time to deploy the adaptive design than the non-adaptive design. To compensate the time consuming issue, the number of tests in the adaptive design should ideally be smaller than the one in the non-adaptive design. 

The second important criterion is between deterministic and randomized procedures. A deterministic procedure is one that produces the same result given the same input, while a randomized procedure includes a degree of randomness in it, and therefore, it might not produce the same result given the same input.  For randomized and possibly adaptive designs, the following categories of randomized procedure are often distinguished:
\begin{itemize}
	\item A \emph{Monte Carlo} procedure outputs the correct result \emph{with high probability} using a \emph{fixed number of tests} which is ideally small.\footnote{Often ``number of tests'' is replaced by ``runtime'' in the context of general algorithms (beyond group testing), but we are primarily interested in the number of tests.}
	\item A \emph{Las Vegas} procedure outputs the correct result with \emph{zero error probability}, but with a \emph{random number of tests} whose average is ideally small.
\end{itemize}
The third important criterion is what is assumed about the defective set, with a particular distinction between probabilistic and combinatorial GT.  The probabilistic setting places a distribution on the items, whereas the combinatorial setting considers the worst case (subject to suitable constraints such as cardinality upper bounds).  We consider the latter scenario, but we note that when we allow randomization and some error probability, this is essentially equivalent to the probabilistic setting by Yao's minimax principle.  Finally, both exact recovery and approximate recovery are often considered in the group testing literature, but the former is much more common and is the focus of our work.

{\bf Overview of literature.} Standard group testing has been intensively studied since its inception~\cite{dorfman1943detection} and widely applied in many fields such as computational and molecular biology~\cite{du2000combinatorial}, networking~\cite{DyachkovPSV19}, and Covid-19 testing~\cite{shental2020efficient,gabrys2020ac}.  In the following, we first outline combinatorial GT with deterministic recovery guarantees, and then we outline probabilistic GT that allows some error probability.  We let $n$ denote the number of items and $d$ the number of defectives. 

For combinatorial GT with non-adaptive designs, it is well-known that a strong $d^2$ dependence is required in the number of tests \cite{kautz1964nonrandom}.  Typical upper bounds provide near-optimal $O(d^2 \log n)$ scaling, and this was attained with $\poly(n)$ decoding time \cite{d1yachkov982bounds,porat11} and then more efficient $\poly(d, \log{n})$ decoding time \cite{indyk2010efficiently,ngo2011efficiently,cheraghchi2013noise,cheraghchi2020combinatorial}.  Moreover, with a two-stage design, de Bonis et al.~\cite{de2005optimal} obtained an order-optimal bound on the number of tests, namely $O(d \log{(n/d)})$.  Refined results for adaptive testing were also given in various works such as \cite{hwang1972method,mezard2011group,aldridge2020conservative}, but since we do not seek to characterize any precise constant factors in this paper, the bound of \cite{de2005optimal} will suffice.

For probabilistic GT, $O(d \log n)$ scaling on the number of tests has long been known, and a line of recent works has led to increasingly precise constant factors in the required number of tests, e.g., see \cite{aldridge2019group,aldridge2014group,scarlett2016phase,coja2020optimal}.  A related line of works reduced the decoding time from $\poly(n)$ to $\poly(d, \ln{n})$ (mostly without seeking precise constant factors), culminating in near-optimal $O(d \log n)$ number of tests \emph{and} decoding time in \cite{cheraghchi2020combinatorial,price2020fast}.

A number of other group testing models have also been considered, with threshold group testing being a notable example \cite{damaschke2006threshold,cheraghchi2013improved,de2017subquadratic,chen2009nonadaptive,bui2021improved,bui2023non}.  However, to our knowledge, almost all of these models are not closely related to our setup; the ones we find to be ``closest'' are outlined in the following subsection.

\subsection{Connection to complex GT and learning a hidden hypergraph}
\label{sub:intro:hiddengraphs}

The problems of complex GT and learning a hidden hypergraph are equivalent; we outlined the former in the introduction, and outline the latter as follows.  The problem consists of learning a hidden hypergraph $G = (V, E)$ using edge-detecting queries~\cite{angluin1988queries,angluin2008learning}, where the number of vertices $|V|= n$ is given to the learner. A test query on a set $S$ of items, denoted as $Q(S)$, returns YES if $S$ contains (as a subset) at least one edge in $E$ and NO otherwise. The goal is to identify $G$ using as few queries as possible.   We re-iterate that ConcGT in a special case of this problem in which $E$ is precisely the set $S_1 \times \dotsc \times S_m$.

In the following, we discuss the combinatorial setting, without any notion of error probability.  
For non-adaptive designs, to make learning possible, it has been shown~\cite{abasi2018non} that $G$ must be a Sperner hypergraph, i.e., no edge is a subset of another. If $G$ has at most $s$ edges and each edge has up to $m$ vertices, the general problem of learning $G$ is equivalent to the problem of exact learning of a monotone disjunctive normal form (MDNF) with at most $s$ monomials in which each monomial contains at most $r$ variables ($s$-term $m$-MDNF) from test queries. The authors also prove that for $n \geq (\max\{s, m \})^2$, any non-adaptive algorithm for learning $s$-term $m$-MDNF must ask at least $\max(t(s-1, m), t(s, m - 1)) \log{n} = \Omega(t(s, m) \log{n})$ test queries and runs in at least $\Omega(t(s, m)n \log{n})$ time in the worst case, where
\begin{equation}
t(s, m) = \frac{s + m}{\log{\binom{s + m}{m} }} \binom{s + m}{m}.
\end{equation}

When $m = 2$, $|S_1| \leq s_1$, and $|S_2| \leq s_2$, $G$ has at most $s = s_1 s_2$ edges. Therefore, the minimum number tests required to solve the ConcGT problem by using the learning a hidden hypergraph approach~\cite{abasi2018non} and the non-adaptive approach with exact recovery is at least
\begin{equation}
t(s, m) \log{n} = t(s_1 s_2, 2) \log{n} = \frac{s_1 s_2 + 2}{\log{\binom{s_1 s_2 + 2}{2} }} \binom{s_1 s_2 + 2}{2} \cdot \log{n} = O\left( \frac{\log{n}}{\log{(s_1 s_2)}} (s_1 s_2)^3 \right). \nonumber
\end{equation}
We will see that this is significantly higher than our ConcGT bound for deterministic non-adaptive designs, whose leading term is $\max\{s_1,s_2\}^3$.  In Appendix \ref{appx:sec:fullCmp}, we perform a more detailed comparison showing similar gains under randomized designs and/or multiple stages of adaptivity; in particular, we will often attain a near-optimal leading term of $\max\{s_1,s_2\}$, compared to $s_1s_2$ or higher for hypergraph learning.  In short, the hidden hypergraph problem (or Complex GT) is significantly more general than ConcGT, but this generality comes at the price of requiring considerably more tests.

\subsection{Contributions}
\label{sub:intro:contri}

Our first contribution is to introduce the ConcGT problem, which has not been considered previously to the best of our knowledge.  Having established the lower bound in Theorem \ref{thm:lower}, we first consider deterministic procedures to find the semi-defective sets for both non-adaptive and adaptive designs.  We assume cardinality upper bounds for the semi-defective subset sizes.  We first state our result for a deterministic adaptive design, which we prove in Section~\ref{sec:det:adaptive}.

\begin{theorem} (Adaptive deterministic design)
Consider the ConcGT problem with $m \geq 2$ disjoint nonempty semi-defective subsets $S_1, \ldots, S_m \subset [n]$ satisfying $|S_i| \leq s_i$ for $i = 1, 2, \ldots, m$. Then:
\begin{itemize}
    \item When $m = 2$, the two semi-defective sets can be found using at most $\log_2{n}$ stages and $O\big( s_1 \log{\frac{n}{s_1}} + s_2 \log{\frac{n}{s_2}} \big)$ tests.
    \item  When $m \geq 2$, all semi-defective sets can be found using at most $\frac{m \log{n}}{1 - \log{2}} + 3$ stages and $O\big( m^2 \log{n} + \sum_{i = 1}^m  s_i \log{\frac{n}{s_i} } \big)$ tests.
\end{itemize}
\label{thm:adaptive:multi}
\end{theorem}

This result amounts to having additional $O(m^2 \log{n})$ scaling in the number of tests compared to the information-theoretic lower bound in Theorem \ref{thm:lower}, for both $m = 2$ and $m \geq 2$.  In particular, we have an order-optimality guarantee when $m = O(1)$, or more generally when the growth of $m$ is sufficiently slow.

Next, we turn to \emph{non-adaptive} deterministic designs,\footnote{Note that here ``deterministic'' means that the test design and decoder are fixed and do not use any randomization.  The existence of these designs is still proved non-constructively using the probabilistic method.}  continuing with the assumption of semi-defective set cardinality upper bounds.  Our results are summarized in the following theorem, and the proof is presented in Section~\ref{sec:det:nonadaptive}.

\begin{theorem} (Non-adaptive deterministic design)
Consider the ConcGT problem with $m=2$ disjoint nonempty semi-defective subsets $S_1, S_2 \subset [n]$, with $|S_1| \leq s_1$ and $|S_2| \leq s_2$. There exists a non-adaptive deterministic algorithm with $O \big( s_{\max}^3 \log{\frac{n}{s_{\max}}} \big)$ tests such that $S_1$ and $S_2$ can be found in $O \big( s_{\max}^3 n^2 \log{\frac{n}{s_{\max}}} \big)$ time, where $s_{\max} = \max \{ s_1, s_2 \}$.
\label{thm:nonadaptive}
\end{theorem}

To reduce the number of tests, we propose randomized procedures with non-adaptive, two-stage, and three-stage designs that recover $S_1$, and $S_2$ with high probability, where $|S_1|$ and $|S_2|$ are known in advance.  In Appendix \ref{app:KnownConst}, we explain how all of these results remain unchanged under the milder assumption that these cardinalities are both known to within constant factors, i.e., the same scaling laws are still maintained but with different hidden constants.  The results are summarized here, and the proofs are presented in Section~\ref{sec:rnd}.

\begin{theorem}(Randomized designs)
Consider the ConcGT problem with $m=2$ disjoint nonempty semi-defective subsets $S_1, S_2 \subset [n]$ satisfying $|S_1| = s_1$ and $|S_2| = s_2$. Let $s_{\max} = \max \{ s_1, s_2 \}$ and $s_{\min} = \min\{ s_1, s_2 \}$.  Then:
\begin{itemize}
	\item[(i)] There exists a non-adaptive Monte Carlo randomized algorithm with $O \big( \frac{s_{\max}^2}{s_{\min}} \log{n} \big)$ tests such that $S_1$ and $S_2$ can be found with probability 0.99 in $O \big( \frac{s_{\max}^2}{s_{\min}} \log{n} \big)$ time.
	\item[(ii)] With two stages, there exists a Monte Carlo randomized algorithm with $O(s_{\max} \log{n})$ tests such that $S_1$ and $S_2$ can be found with probability 0.99 in $O(s_{\max} \log{n})$ time. 
	\item[(iii)]  With three stages, there exists a Monte Carlo randomized algorithm with $O(s_{\max} \log{\frac{n}{s_{\max}}})$ tests such that $S_1$ and $S_2$ can be found with probability 0.99 in $O\big( s_{\max} \log^2{\frac{n}{s_{\max}}} \big)$ time.
	\item[(iv)] There exists a Las Vegas design attaining zero error probability with an \emph{average} number of stages $3+\eta$ for arbitrarily small $\eta > 0$, and an \emph{average} number of tests and runtime given by $O(s_{\max} \log{\frac{n}{s_{\max}}})$ and $O\big( s_{\max} \log^2{\frac{n}{s_{\max}}} \big)$ respectively.
\end{itemize}
\label{thm:rnd}
\end{theorem}

We note that success probability 0.99 can be replaced by any constant in $(0,1)$, whereas when the error probability is $1-\epsilon$ and the dependence on $\epsilon$ is sought, certain terms need to be multiplied by $\log\frac{1}{\epsilon}$.  This is made precise in the proof in Section~\ref{sec:rnd}, but we omit it from the theorem statement.  In the following subsection, we compare Theorem \ref{thm:rnd} to the lower bound in Theorem \ref{thm:lower}, establishing order-optimality in many cases of interest.

\subsection{Comparisons}
\label{sub:intro:cmp}

We summarize our results in Table~\ref{tbl:cmp}; we also provide an expanded version in Appendix~\ref{appx:sec:fullCmp} where we demonstrate the suboptimality of specializing results from hypergraph learning (e.g.,~\cite{chen2009nonadaptive,abasi2019learning}) to our problem.  For example, for non-adaptive deterministic designs with $m=2$, our number of tests is at least a factor $s_{\min}^3$ less than a direct application of~\cite{chen2009nonadaptive}.  Despite this improvement, we see that the degree of optimality of our non-adaptive deterministic designs remains unclear (and an interesting open problem).  In contrast, for deterministic adaptive designs, we match the lower bound when $m$ is constant, and more generally match it to within an additive  $O(m^2 \log{n})$ term.

Next, we consider the upper bounds for randomized algorithms in Theorem \ref{thm:rnd}, and compare them to the lower bound in Theorem \ref{thm:lower}.  We see that even in the non-adaptive case, we have order-optimality in broad cases of interest, namely, when $s_{\max} = \Theta(s_{\min})$ and $\log\frac{n}{s_{\max}} = \Theta(\log n)$.  The former of these is perhaps the less mild assumption, and we see that it can be dropped upon moving to a two-stage design.  With three stages (or $3+\eta$ on average with a Las Vegas design), we additionally obtain the correct logarithmic factor, meaning that we match the lower bound with \emph{neither} of the preceding restrictions.

\begin{table}[]
\centering
\begin{tabular}{|c|c|l|c|c|}
\hline
\begin{tabular}[c]{@{}c@{}}Design\\ type\end{tabular} & No. of stages & \multicolumn{1}{c|}{\begin{tabular}[c]{@{}c@{}}Semi-defective\\ sets\end{tabular}} & \multicolumn{1}{c|}{Scheme} & No. of tests \\
\hline
\multicolumn{1}{|l|}{Arbitrary} & \multicolumn{1}{c|}{Any} & \begin{tabular}[c]{@{}c@{}} $|P_i| \leq s_i$ for \\ $i = 1, \ldots, m$ \end{tabular} & \begin{tabular}[c]{@{}c@{}} Any (Lower bound) \\\textbf{Theorem~\ref{thm:lower}} \end{tabular} & \multicolumn{1}{l|}{\begin{tabular}[c]{@{}l@{}} $\Omega \big( \sum_{i = 1}^m s_i \log{\frac{n}{s_i}} \big)$ if $\sum_{v = 1}^{m} s_v \leq \frac{n}{2}$ \end{tabular}} \\
\hline
\multirow{3}{*}{\begin{tabular}[c]{@{}l@{}} \\ \\ Deterministic \end{tabular}} & 1 & \multirow{2}{*}{\begin{tabular}[c]{@{}l@{}} $|P_1| \leq s_1$ \\ $|P_2| \leq s_2$ \end{tabular}} & \textbf{Theorem~\ref{thm:nonadaptive}} & $O \left( s_{\max}^3 \log{\frac{n}{s_{\max}}} \right)$ \\
\cline{2-2} \cline{4-5} & $\log_2{n}$ & & \multirow{4}{*}{\textbf{Theorem}~\ref{thm:adaptive:multi}} & $O\left( s_1 \log{\frac{n}{s_1}} + s_2 \log{\frac{n}{s_2}} \right)$ \\
\cline{2-3} \cline{5-5} & $\frac{m \log{n}}{1 - \log{2}} + 3$ & \begin{tabular}[c]{@{}c@{}} $|P_i| \leq s_i$ for \\ $i = 1, \ldots, m$ \end{tabular} & & $O\left( m^2 \log{n} + \sum_{i = 1}^m  s_i \log{\frac{n }{s_i} } \right)$ \\
\hline
\multirow{3}{*}{Random} & 1 & \multirow{3}{*}{\begin{tabular}[c]{@{}l@{}} $|P_1| = s_1$ \\ $|P_2|= s_2$ \end{tabular}} & \multirow{3}{*}{\begin{tabular}[c]{@{}l@{}} \textbf{Theorem}~\ref{thm:rnd}\end{tabular}} & \begin{tabular}[c]{@{}l@{}} $O \big( \frac{s_{\max}^2}{s_{\min}} \log n \big)$ \end{tabular} \\ 
\cline{2-2} \cline{5-5} & 2 & & & \begin{tabular}[c]{@{}l@{}} $O( s_{\max} \log{n} )$ \end{tabular} \\
\cline{2-2} \cline{5-5} & 3 (Monte Carlo); avg.~$3+\eta$~(Las Vegas) & & & \begin{tabular}[c]{@{}l@{}} $O \big( s_{\max} \log{\frac{n}{s_{\max}}} \big)$ \end{tabular} \\
\hline
\end{tabular}

\caption{Comparison of proposed schemes with previous ones. Given a population of $n$ items, there are $m \geq 2$ disjoint nonempty semi-defective subsets $S_1, \ldots, S_m \subset [n]$ with $|S_i| \leq s_i$ for $i = 1, 2, \ldots, m$. When considering randomized design, we set $m = 2$, $|S_1| = s_1$, and $|S_2| = s_2$ for our proposed schemes. In the bottom row, $\eta> 0$ is an arbitrarily small constant.}
\label{tbl:cmp}
\end{table}

\subsection{Overview of Techniques} \label{sec:overview}

With the exception of our non-adaptive deterministic design (which is based on a suitable generalization of the \emph{disjunct} property \cite{du2000combinatorial}), our algorithms share the common idea of performing the following two steps:
\begin{itemize}
    \item[(i)] Identify a subset that intersects $S_1$ but not $S_2$, or vice versa (when $m=2$).  For $m > 2$, this generalizes to intersecting exactly one of the sets $S_1,\dotsc,S_m$, and for adaptive designs these ``subsets'' will in fact be singletons.
    \item[(ii)] Observe that when $m=2$, if we perform a series of tests that always contain the subset intersecting $S_1$ (identified in step (i)), then each test will be 1 if and only if the test contains at least one item from $S_2$.  Thus, we have a \emph{standard group testing problem} for identifying $S_2$.  A similar argument applies when the roles of $S_1$ and $S_2$ are reversed, and also when $m > 2$ and a suitable union of $m-1$ subsets is always included in the test.
\end{itemize}
Once step (i) is achieved, step (ii) is often straightforward, since we can use existing group testing algorithms (e.g., non-adaptive, two-stage, or fully adaptive) in a black-box manner.  In the non-adaptive setting, however, difficulties arise from not knowing the subset intersecting $S_1$ only (or $S_2$ only) in advance; we overcome this by considering multiple instances of standard group testing in parallel, and identifying the appropriate one at decoding time.

For step (i) above, we identify the subsets of interest using different methods depending on the amount of adaptivity.  For the deterministic adaptive design, we use a splitting idea that bears some similarity to binary splitting in regular GT \cite{hwang1972method}.  For our randomized designs, the rough idea is to form random subsets of $\{1,\dotsc,n\}$ of suitable size such that the probability of a desired event (e.g., ``the subset contains exactly one item from $S_1 \cup S_2$'') is constant.  By repeating this procedure independently $O\big(\log\frac{1}{\epsilon'}\big)$ times for some $\epsilon' > 0$, the desired event then occurs at least once with probability at least $1-\epsilon'$.  In addition, we design the procedure in a manner that ensures that when the desired subset exists, it can easily be identified.  The details will be given in the subsequent sections.

\section{Deterministic Designs}
\label{sec:det}

\subsection{Non-adaptive design}
\label{sec:det:nonadaptive}

In this subsection, we present the proof of Theorem~\ref{thm:nonadaptive} for deterministic non-adaptive designs.  Recall that $m = 2$ with $|S_1| \leq s_1$ and $|S_2| \leq s_2$. Since the roles of $s_1$ and $s_2$ are symmetric, without loss of generality, we assume $s_1 \leq s_2$, i.e, $s_{\max} = \max \{ s_1, s_2 \} = s_2$ and $s_{\min} = \min \{s_1, s_2 \} = s_1$.

Before going to our proposed algorithm, we first introduce $(n, u, v)$-disjunct matrices, which are a useful combinatorial tool to deal with several models in group testing \cite{du2000combinatorial,chen2008upper}. The support set of a vector $\bX = (x_1, x_2, \ldots, x_p)$ is defined as $\supp(\bX) = \{ j \in \{1, \ldots, p \} \mid x_j \neq 0 \}$, and we let $\cM(:, j)$ denote the $j$th column of matrix $\cM$. A binary matrix $\cM$ is \emph{$(n, u, v)$-disjunct} if for any $u + v$ columns $j_1, j_2, \ldots, j_{u + v}$, we have
\begin{equation}
\left\vert \bigcup_{i = 1}^u \supp( \cM(:, j_i)) \setminus \bigcup_{h = u + 1}^{u + v} \supp(\cM(:, j_h)) \right\vert \geq 1,
\end{equation}
i.e., for any $u + v$ columns there exists at least one row where any $u$ designated columns have 1-entries and the other $v$ columns have 0-entries.

Our algorithm makes use of an $(n, 2, s_{\max})$-disjunct matrix $X$, and we define $t_0(X)$ to be the number of negative tests in which all items in $X$ appear. The decoding algorithm for finding the sets $S_1$ and $S_2$ is described in Algorithm~\ref{alg:nonadaptive:2}.

\begin{algorithm}
\caption{Non-adaptive algorithm for finding two semi-defective sets in ConcGT.}
\label{alg:nonadaptive:2}
\noindent \textbf{Input:} The set of input items $N$, the maximum cardinalities of two semi-defective sets $s_1$ and $s_2$, a $t \times n$ $(n, 2, s_{\max})$-disjunct matrix $X$, where $s_{\max} = \max \{ s_1, s_2 \}$.\\
\textbf{Output:} The two semi-defective sets.

\begin{algorithmic}[1]
\State Let $S$ be the set of all pairs $X$ such that $t_0(X) = 0$. \label{alg:nonadaptive:2:candidates}
\State Pick arbitrarily a pair $(x_1, x_2) \in S$. Return $\hat{S}_1 = \{ x \in N \mid (x, x_2) \in S \}$ and $\hat{S}_2 = \{ x \in N \mid (x_1, x) \in S \}$. \label{alg:nonadaptive:2:separation}
\end{algorithmic}
\end{algorithm}

\subsubsection{Correctness}
\label{sec:det:Correctness}

Since the input matrix is $(n, 2, s_{\max})$-disjunct, for any two indices $x_1$ and $x_2$ such that $(x_1, x_2)$ is not a defective pair, there exists a row such that $x_1$ and $x_2$ belong to that row and that row has empty intersection with at least one of $S_1$ and $S_2$. Therefore, the test on that row returns negative. This implies that $t_0( \{x_1, x_2 \} ) \geq 1$ and that test reveals the non-defective pair $(x_1, x_2)$. On the other hand, any test containing a defective pair returns positive, i.e., $t_0( \{x_1, x_2 \} ) = 0$. Step~\ref{alg:nonadaptive:2:candidates} therefore returns the set $S$ of all defective pairs.

Since a test on a subset of items returns positive if the subset contains at least one item in each semi-defective set and $S$ contains all defective pairs, Step~\ref{alg:nonadaptive:2:separation} simply separates all items in $S$ into two subsets which are also the two semi-defective sets.

\subsubsection{Complexity}
\label{sec:det:Complexity}

The number of rows (tests) of an $(n, 2, s_{\max})$-disjunct matrix can be bounded~\cite{chen2008upper} by
\begin{equation}
t^* = \left( \frac{s_{\max} + 2}{2} \right)^2 \left( \frac{s_{\max} + 2}{s_{\max}} \right)^{s_{\max}} \left(1 + (s_{\max} + 2) \left( 1 + \log{\left( \frac{n}{s_{\max} + 2} + 1 \right)} \right) \right) = O \left( s_{\max}^3 \log{\frac{n}{s_{\max}}} \right).
\end{equation} 

Since there are $\binom{n}{2} = O(n^2)$ pairs, it takes at most $O(t^* n^2) = O \big( s_{\max}^3 n^2 \log{\frac{n}{s_{\max}}} \big)$ time to execute Step~\ref{alg:nonadaptive:2:candidates}. Meanwhile, since $|S_1| \leq s_1$ and $|S_2| \leq s_2$, there are up to $s_1 s_2$ defective pairs, and hence $|S| \leq s_1 s_2$. Consequently, it takes at most $2 s_1 s_2$ time to complete Step~\ref{alg:nonadaptive:2:separation}. Therefore, the decoding complexity of Algorithm~\ref{alg:nonadaptive:2} is
\begin{equation}
O \left( s_{\max}^3 n^2 \log{\frac{n}{s_{\max}}} \right) + O(s_1 s_2) = O \left( s_{\max}^3 n^2 \log{\frac{n}{s_{\max}}} \right).
\end{equation}

\subsection{Adaptive design}
\label{sec:det:adaptive}

In this section, we present the proof of Theorem~\ref{thm:adaptive:multi} for deterministic adaptive designs.  We consider a multi-stage approach, in which the design of a test may now depend on the test results from the previous stages. Here we consider the general case of $m \ge 2$.

The general idea behind our adaptive design follows the outline given in Section \ref{sec:overview}: We identify a defective tuple $D_0 = \{ d_1, d_2, \ldots, d_m \}$, and then one-by-one for each $i=1,\dotsc,m$, we perform a sequence of tests that all contain $D_0 \setminus \{d_i\}$, allowing us to identify $S_i$ via standard group testing techniques.  For the latter step, in general notion, we let $\mathrm{Adaptive}(\leq u, A, B )$ denote an adaptive strategy that returns a semi-defective set with the size of up to $u$ in $A$; every test in this procedure contains a subset of items of $A$ and all items in $B$ (which represents $D_0 \setminus \{d_i\}$ in the above discussion).  For convenience, we will specifically use the 2-stage design for standard group testing in \cite[Corollary 3]{de2005optimal}, as it comes with a convenient non-asymptotic bound on the number of tests of $7.54 u \log_2{\frac{|A|}{u}} + 16.21 u - 2 \mathrm{e} - 1 = O \big( u \log{\frac{|A|}{u}} \big)$.

In more detail, the first step is to find a set consisting of a single defective tuple, denoted by $D_0 = \{ d_1, d_2, \ldots, d_m \}$. The second step is to use $D_0$ for recovering all semi-defective sets. Specifically, the defective set $S_i$ can be found by setting $S_i = \{ d_i \} \cup \mathrm{Adaptive}(\leq s_i - 1, N \setminus D_0, D_i )$, where $D_i = D_0 \setminus \{ d_i \}$ for $i = 1, \ldots, m$. The numbers of stages and tests for finding $S_i$ are two and up to $7.54 (s_i - 1) \log_2{\frac{n - m - \sum_{v = 1}^{i - 1} s_i}{s_i - 1}} + 16.21 (s_i - 1) - 2 \mathrm{e} - 1$, respectively. By running these 2-stage standard GT subroutines for $i=1,\dotsc,m$ in parallel, the number of stages in the second step is $2$ and the number of tests is up to
\begin{equation}
\sum_{i = 1}^m \bigg( 7.54 (s_i - 1) \log_2{\frac{n - m}{s_i}} + 16.21 (s_i - 1) - 2 \mathrm{e} - 1 \bigg) = O \left( \sum_{i = 1}^m  s_i \log{\frac{n}{s_i}} \right). \label{eqn:totalAdaptive:tests}
\end{equation}
Since the designs and their analyses in the second step are directly from~\cite[Corollary 3]{de2005optimal}, we only need to design and analyze the first step.

\subsubsection{The case $m = 2$}
\label{sub:adaptive:2}

We present an algorithm to tackle the case $m = 2$ in Algorithm~\ref{alg:Adaptive:2}.  We wish to show that a defective pair is returned in Step~\ref{alg:Adaptive:2:Step1}.  
To see this, we first argue that in Step~\ref{alg:Adaptive:2:Init}b, there must exist a test that returns positive: If either $A_{01}$ or $A_{02}$ contains items from both sets $S_1$ and $S_2$, the test on that subset returns positive. Otherwise, one must contain $S_1$ and the other contains $S_2$.  In this scenario, without loss of generality, assume that $S_1$ and $S_2$ are subsets of $A_{01}$ and $A_{02}$, respectively. By dividing $A_{01}$ into two halves $A_{11}$ and $A_{12}$, there must exist a half that contains at least one item in $S_1$. Similarly, there must exist a half of $A_{02}$ that contains at least one item in $S_2$. Consequently, there exists a subset among $A_{11} \cup A_{21}, A_{11} \cup A_{22}, A_{12} \cup A_{21}, A_{12} \cup A_{22}$ such that the test on it returns positive.

\begin{algorithm}
\caption{Adaptive algorithm for finding two semi-defective sets in ConcGT.}
\label{alg:Adaptive:2}
\noindent \textbf{Input:} The set of input items $N$, the cardinalities of two semi-defective sets $S_1$ and $S_2$, namely $s_1$ and $s_2$.\\
\textbf{Output:} Two semi-defective sets.

\begin{algorithmic}[1]
\State Find an item in $S_1$ and an item in $S_2$. \label{alg:Adaptive:2:Step1}
	\begin{algsubstates}
	\State Let $A_{01}$ and $A_{02}$ be the first half and the second half of the set $A$ (initialized to $A = N$). Let $A_{11}$ and $A_{12}$ be the first half and the second half of $A_{01}$. Similarly, let $A_{21}$ and $A_{22}$ be the first half and the second half of $A_{02}$. \label{alg:Adaptive:2:Init}
	\State Create 6 tests on the subsets $A_{01}, A_{02}, A_{11} \cup A_{21}, A_{11} \cup A_{22}, A_{12} \cup A_{21}, A_{12} \cup A_{22}$. \label{alg:Adaptive:2:Divide}
	\State Replace $A$ in Step~\ref{alg:Adaptive:2:Init}a by a subset that returns positive in Step~\ref{alg:Adaptive:2:Divide}b. \label{alg:Adaptive:2:Pick}
    \State Repeat Steps~\ref{alg:Adaptive:2:Init}a to~\ref{alg:Adaptive:2:Stop}d until the subset that returns positive in Step~\ref{alg:Adaptive:2:Divide}b has size 2. \label{alg:Adaptive:2:Stop}
    \end{algsubstates}
\State Recover $S_1$ and $S_2$: Let $a_1$ and $a_2$ be the two items returned in the previous step, and estimate $\hat{S}_i = \{ a_i \} \cup \mathrm{Adaptive}(\leq s_i - 1, N \setminus \{ a_1, a_2 \}, \{ a_{i^c} \} )$ for $i = 1, 2$, where $a_{1^c} := a_2$ and $a_{2^c} := a_1$. \label{alg:Adaptive:2:Step2}
\end{algorithmic}
\end{algorithm}

We now move to the complexity analysis. Each subset in Step~\ref{alg:Adaptive:2:Divide}b has size at least $\ceil{|A|/2}$ and at most $\floor{|A|/2} + 1$. Since Steps~\ref{alg:Adaptive:2:Init}a to~\ref{alg:Adaptive:2:Stop}d are repeated until the size of $|A|$ is 3 or 4, i.e., the stop condition in Step~\ref{alg:Adaptive:2:Stop}d) holds, there are up to $\ceil{\log_2{n}} - 2$ stages and $6(\ceil{\log_2{n}} - 2)$ tests in total for the first step. As analyzed in the beginning of Section~\ref{sec:det:adaptive}, in the second step, the number of stages is $2$ and the number of tests is
\begin{equation}
\sum_{i = 1}^2  7.54 (s_i - 1) \log_2{\frac{n - m}{s_i - 1}} + 16.21 (s_i - 1) - 2 \mathrm{e} - 1.
\end{equation}
Therefore, there are up to $\log_2{n} - 2 + 2 = \log_2{n}$ stages and 
\begin{align}
& 6\ceil{\log_2{n}} - 12 + \sum_{i = 1}^2 \left(  7.54 (s_i - 1) \log_2{\frac{n - m}{s_i - 1}} + 16.21 (s_i - 1) - 2 \mathrm{e} - 1 \right) \\
&= O \left( \sum_{i = 1}^2  s_i \log{\frac{n}{s_i}} \right)
\end{align}
tests in total.

\subsubsection{The case $m \geq 2$}
\label{sub:adaptive:multi}

For the case $m \geq 2$, we present Algorithm~\ref{alg:Adaptive:multi}.  This algorithm is in fact slightly simpler to describe than Algorithm \ref{alg:Adaptive:2}, but it requires more stages of adaptivity (even when $m=2$);  hence, we found it useful to provide both.  
Similar to the arguments in Section~\ref{sub:adaptive:2}, since the correctness of Step~\ref{alg:Adaptive:m:Step2} is already proved in the beginning of Section~\ref{sec:det:adaptive}, we only need to prove the correctness of the first step in Algorithm~\ref{alg:Adaptive:multi}.

\begin{algorithm}
\caption{Adaptive algorithm for finding $m$ semi-defective sets in ConcGT.}
\label{alg:Adaptive:multi}
\noindent \textbf{Input:} The set of input items $N$, the cardinalities of $m$ semi-defective sets $s_1, s_2, \ldots, s_m$.\\
\textbf{Output:} The $m$ semi-defective sets.

\begin{algorithmic}[1]
\State Find an item in each $S_i$ for $i = 1, 2, \ldots, m$. \label{alg:Adaptive:m:Step1}
	\begin{algsubstates}
	\State Divide $A = N$ into $m + 1$ subsets of equal size (or up to a difference of $1$ due to rounding). Test the $m + 1$ pools consisting of $m$ subsets of these subsets.  \label{alg:Adaptive:m:Init}
	\State Set $A$ to be a subset that returns positive in Step~\ref{alg:Adaptive:m:Init}a. \label{alg:Adaptive:m:Pick}
    \State Repeat Steps~\ref{alg:Adaptive:m:Init}a to~\ref{alg:Adaptive:m:Stop}c until the subset returns positive in Step~\ref{alg:Adaptive:m:Pick}b has size of $m$. \label{alg:Adaptive:m:Stop}
    \end{algsubstates}
\State Recover $S_i$ for $i = 1, 2, \ldots, m$. \label{alg:Adaptive:m:Step2}
	\begin{algsubstates}
	\State Let $D_0 = \{d_1, \ldots, d_m \}$ be the subset returned once the first step stops. 
	\State Set $\hat{S}_i = \{ d_i \} \cup \mathrm{Adaptive}(\leq s_i - 1, N \setminus D_0, D_0 \setminus \{ d_i \} )$ for $i = 1, 2, \ldots, m$.
	\end{algsubstates}
\end{algorithmic}
\end{algorithm}

The procedure in the first step is inspired by the adaptive procedure in threshold group testing~\cite{damaschke2006threshold}.  We claim that there must exist a pool that returns positive in Step~\ref{alg:Adaptive:m:Step1}a. Indeed, on the first iteration with $A=N$, there are $m$ semi-defective sets $S_1, \ldots, S_m$ and they are distributed into $m + 1$ subsets of $N$, so there must exist a union of $m$ pools such that the union contains items from all of $S_1,\dotsc,S_m$, based on the pigeonhole principle.  Since the algorithm does not discard any entire semi-defective set in Step~\ref{alg:Adaptive:m:Step1}b, the same reasoning applies on subsequent iterations with $A \subset N$.

Algorithm~\ref{alg:Adaptive:multi} requires $\log_{(m + 1)/m}{n} + 1 \leq \frac{m \log{n}}{1 + \log{(1 - 1/m)}} + 1 \leq \frac{m \log{n}}{1 - \log{2}} + 1$ stages and $(m + 1)\log_{(m + 1)/m}{n} + 1 \leq \frac{m (m + 1) \log{n}}{1 + \log{(1 - 1/m)}} + 1 \leq \frac{m(m + 1) \log{n}}{1 - \log{2}} + 1$ tests to execute the first step. Then, by using the analysis in the beginning of Section~\ref{sec:det:adaptive}, the total number of stages is up to $\frac{m \log{n}}{1 - \log{2}} + 1 + 2 = \frac{m \log{n}}{1 - \log{2}} + 3$, and the total number of tests is up to
\begin{align}
&\frac{m(m + 1) \log{n}}{1 - \log{2}} + 1 + \sum_{i = 1}^m \left(  7.54 (s_i - 1) \log_2{\frac{n - m}{s_i - 1}} + 16.21 (s_i - 1) - 2 \mathrm{e} - 1 \right) \\
&= O\left( m^2 \log{n} + \sum_{i = 1}^m  s_i \log{\frac{n}{s_i}} \right).
\end{align}

\section{Randomized Designs}
\label{sec:rnd}

In this section, we propose Monte Carlo randomized designs in the categories of non-adaptive, two-stage, and three stage, such that the two semi-defective sets $S_1$ and $S_2$ can be recovered with high probability.  We also provide a Las Vegas design with $3+\eta$ (e.g., $3.01$) stages \emph{on average}, which follows easily from the 3-stage Monte Carlo design.  
The results are summarized in Theorem~\ref{thm:rnd}.  All of our designs in this section are broadly based on the outline given in Section \ref{sec:overview}, in which we first identify subsets that intersect $S_1$ but not $S_2$ (or vice versa), and then apply standard group testing techniques.

\textbf{Notations:} Let the cardinalities of two semi-defective sets $S_1$ and $S_2$ be $s_1$ and $s_2$, respectively.  We assume that these are exactly specified and known to the test designer throughout this section, but in Appendix \ref{app:KnownConst} we relax this requirement to only requiring knowledge of their cardinalities to within a constant factor.

Set $S = S_1 \cup S_2$, $S_0 = [n] \setminus S$ and $s_0 = |S_0| = n - s_1 - s_2$. Let $\test(X)$ be a test on the subset $X$. We introduce the following notation for testing with respect to a given matrix $\cM$ but with the items from a given set $A$ additionally included in every test:
\begin{equation}
\test(\cM, A) = \begin{bmatrix}
\test(\supp(\cM(1, :) \cup A )) \\
\vdots \\
\test(\supp(\cM(t, :) \cup A )) \\
\end{bmatrix}.
\end{equation}
This means that test $i$, i.e., $\test(\supp(\cM(i, :) \cup A )$, is positive if $|\supp(\cM(i, :) \cup A) \cap S_1 | \geq 1$ and $|\supp(\cM(i, :) \cup A) \cap S_2 | \geq 1$, and is negative otherwise. When $A = \emptyset$, we simply recover $\test(\cM, \emptyset) = \test(\cM)$.

Let $\bX_i = (x_1^i, \ldots, x_n^i)$ be the representation vector of the input set $N$ with respect to the semi-defective set $S_i$, where $x_j^i = 1$ if $j \in S_i$ and $x_j^i = 0$ otherwise. Then, we get
\begin{equation}
\test(\cM, A) = \begin{cases}
\mathbb{1}, \qquad \quad \ \text{if } A \cap S_1 \neq \emptyset \text{ and } A \cap S_2 \neq \emptyset. \\
\cM \odot \bX_1, \ \text{ if } A \cap S_1 = \emptyset \text{ and } A \cap S_2 \neq \emptyset. \\ 
\cM \odot \bX_2, \ \text{ if } A \cap S_1 \neq \emptyset \text{ and } A \cap S_2 = \emptyset. \\ 
\test(\cM), \ \text{if } A \cap S_1 = \emptyset \text{ and } A \cap S_2 = \emptyset,
\end{cases}
\label{eqn:classicalConversion}
\end{equation}
where $\cM \odot \bX$ denotes \emph{standard group testing} on a measurement matrix $\cM$ with defectivity indicator vector $\bX$, i.e.:
\begin{itemize}
    \item If the $i$-th row of $\cM$ contains a $1$ in any entry where $\bX$ is $1$, then the $i$-th entry of $\cM \odot \bX$ is $1$;
    \item If the $i$-th row of $\cM$ contains a $0$ in all entries where $\bX$ is $1$, then the $i$-th entry of $\cM \odot \bX$ is $0$.
\end{itemize}

Let $\mathrm{Decode}(\bZ = \cM \odot \bX, \cM)$ be a standard GT decoding procedure for recovering a vector $\bX$ with $|\bX| \leq s_{\max}$, where $s_{\max} = \max\{s_1, s_2 \}$.  We will use different such procedures depending on the number of stages we are targeting, and in some cases we will allow some error probability in this procedure.

\textbf{Random sampling subroutine:} Here we present an important building block of our randomized designs. Following the high-level ideal outlined in Section \ref{sec:overview}, to identify $S_1$, we take the following approach:
\begin{enumerate}
\item Find a subset $A$ of $N$ such that it contains at least one item in $S_2$ and none in $S_1$.
\item Run $\mathrm{Decode}(\cM \odot \bX_1 = \test(\cM, A), \cM)$ which is a non-adaptive strategy that returns a semi-defective set in $N$ with size up to $s_1$.
\end{enumerate}
Note that the equality $\cM \odot \bX_1 = \test(\cM, A)$ in Step 2 follows directly from the second case in \eqref{eqn:classicalConversion}.  Similarly, to identify $S_2$, one first finds a subset $B$ of $N$ such that it contains at least one item in $S_1$ and none in $S_2$. Then we call $\mathrm{Decode}(\cM \odot \bX_2 = \test(\cM, A), \cM)$ to get $S_2$.

Since the procedure $\mathrm{Decode}(\bZ = \cM \odot \bX, \cM)$ for $|\bX| \leq s_{\max}$ can easily be characterized using known group testing results, our main task is to ensure that finding $A$ (or $B$) satisfies the above requirements, and determine how to execute both steps simultaneously to attain a non-adaptive strategy. The following lemma establishes the requirement of finding $A$ and $B$, and is proved in Appendix~\ref{appx:sec:lem}.

\begin{lemma}
Given a population of $n$ items, suppose that there are two disjoint nonempty subsets $S_1$ and $S_2$ with $|S_1| = s_1$ and $|S_2| = s_2$. Let $s = s_1 + s_2$. Consider sampling $n/s$ items without replacement $t_i$ times for $i = 1, 2$. Given a precision parameter $\epsilon_i > 0$, the probability that there exists a trial (among the $t_i$ performed) with one item from $S_i$ and none from $S_{i^c}$ (i.e., the other set in $\{S_1,S_2\}$) is at least $1-\epsilon_i$ provided that $t_i \geq \frac{\mathrm{e}^5}{4 \pi^2} \log{\frac{1}{\epsilon_i}} \cdot \frac{s}{s_i}$.
\label{lem:S1S2}
\end{lemma}

We now proceed with the details of our non-adaptive design, and then our 2-stage, 3-stage, and Las Vegas designs.

\subsection{Non-adaptive design}
\label{sec:rnd:1Stage}

\subsubsection{Encoding procedure}
\label{sec:rnd:1Stage:Enc}

There are two phases, which we call initialization and testing, illustrated in Fig~\ref{fig:Enc_Rand}. The testing phase consists of three sub-phases to generate three types of vectors: indicator vectors, reference vectors, and identification vectors.  Note that the reference vector can be interpreted as creating $t_{\cA} t_{\cB}$ tests $\test(A_i, B_j)$, where $A_i = \supp(\cA(i, :))$ and $B_j = \supp(\cB(j, :))$ for $i \in [t_{\cA}]$ and $j \in [t_{\cB}]$, as depicted via the double-arrow in Fig.~\ref{fig:Enc_Rand}. 

The indicator vectors, which are $\bY_{\cA}$ and $\bY_{\cB}$, are first used to identify tests in $\cA$ and $\cB$ that contain at least one item from \emph{each} (i.e., both) of $S_1$ and $S_2$.  Such tests are not used by our algorithm, so they are removed from consideration.  After doing so, our next task is to eliminate tests in $\cA$ and $\cB$ that contain no items from $S_1 \cup S_2$. This is done by using a reference vector $\bR_{\cA}$ as described below. After these two rounds of eliminating tests, the remaining tests in $\cA$ and $\cB$ must contain only items in $S_1$ or only items in $S_2$. By using these tests with the identification vectors $\bF_{\cA}$ and $\bF_{\cB}$, one can recover $S_1$ and $S_2$ as described below.

We now proceed to describe the choices of the relevant parameters, matrices, etc.  
Let $\epsilon_1$ and $\epsilon_2$ be positive precision constants. By using Lemma~\ref{lem:S1S2}, we can let $\cA$ and $\cB$ be two binary matrices of size $t_{\cA} \times n$ and $t_{\cB} \times n$ with $t_{\cA} = t_1$ and $t_{\cB} = t_2$, where each row contains $n/s$ items from $n$ items, respectively.

Let $\cM$ be an $m \times n$ matrix such that for any $\bX \in \{0 ,1 \}^n$ and $|\bX| \leq s_{\max}$, given $\bY = \cM \odot \bX$, $\bX$ can be recovered in time $O(m)$ with high probability. Such a matrix can be randomly obtained with $m = O(s_{\max} \log{n})$ and error probability \cite{cheraghchi2020combinatorial}
\begin{equation}
    \epsilon_0 = \mathrm{e}^{-\Omega(s_{\max}) } + \frac{1}{\poly(n)}. \label{eq:eps0}
\end{equation}
In particular, in the case that $s_{\max} \to \infty$ as $n \to \infty$, we have $\epsilon_0 \to 0$, and thus $\epsilon_0 < \min\{\epsilon_1,\epsilon_2\}$ when $n$ is large enough (note that $\epsilon_1$ and $\epsilon_2$ are constant).  If $s_{\max} = O(1)$, then the result in~\cite{cheraghchi2020combinatorial} does not give $\epsilon_0 \to 0$, but the proof still readily shows that $\epsilon_0$ can be made arbitrarily small by using a large enough implied constant in $m = O(s_{\max} \log{n})$, so we still can still ensure $\epsilon_0 \le \min\{\epsilon_1,\epsilon_2\}$.\footnote{Alternatively, we could use the scheme proposed in \cite[Appendix B]{bondorf2020sublinear} that attains $\epsilon_0 \to 0$, at the expense of increasing the decoding time from $O(\log n)$ to $O(\log ^2 n)$.}  Subsequently, we set $\epsilon = \frac{1}{4}\max\{\epsilon_0, \epsilon_1, \epsilon_2 \}$, so that $2\epsilon_0 + \epsilon_1 + \epsilon_2 \le \epsilon$ (here the factor of $2$ is due to decoding twice, once for each of $S_1$ and $S_2$).

\begin{figure}
\centering
1\includegraphics[scale=0.5]{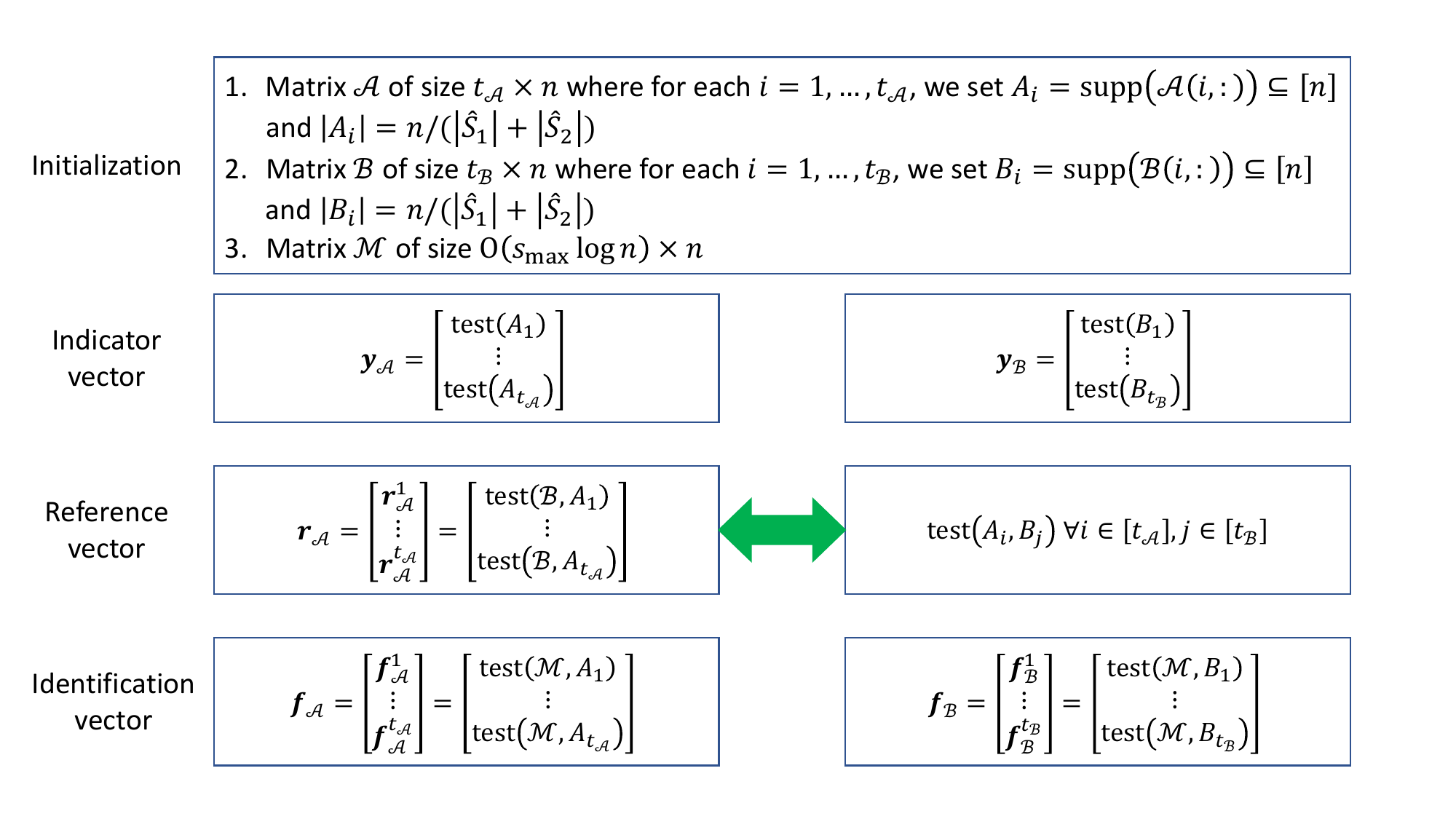}
\caption{Encoding procedure for the non-adaptive and randomized design. There are multiple phases that produce three types of vectors: indicator vectors, reference vectors, and identification vectors. The rows of $\cA$ and $\cB$ are generated using uniform sampling without replacement.}
\label{fig:Enc_Rand}
\end{figure}

\subsubsection{Decoding and correctness}
\label{sec:rnd:1Stage:DecCorrect}

\begin{figure}
\centering
\includegraphics[scale=0.5]{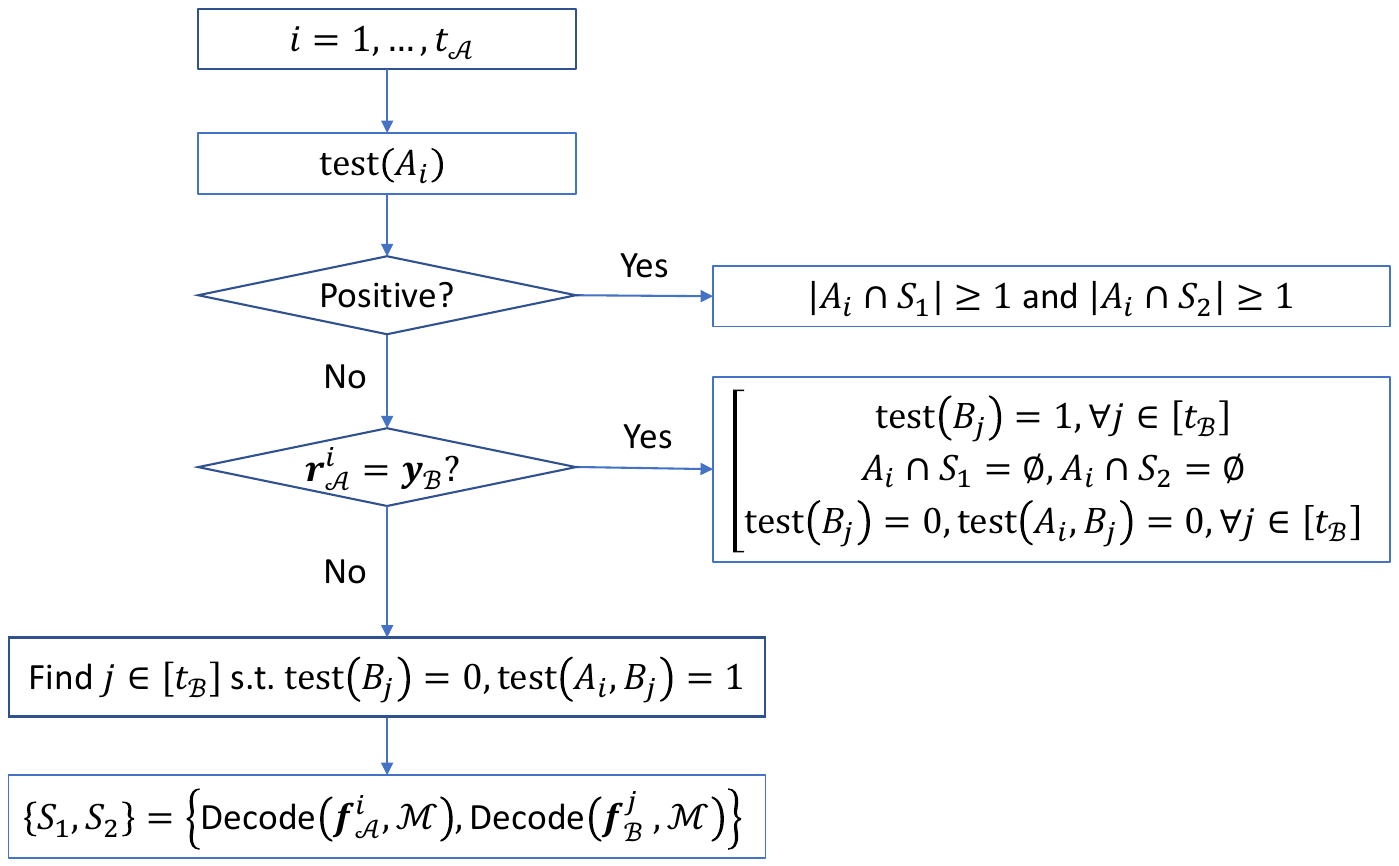}
\caption{Decoding procedure for the non-adaptive and randomized design.}
\label{fig:Dec_Rand}
\end{figure}

The decoding procedure is illustrated in Fig~\ref{fig:Dec_Rand}. We recover $S_1$ and $S_2$ by examining the matrix $\cA$. For every $i = 1, \ldots, t_{\cA}$, we examine whether the test outcome $\test(A_i)$ is positive, where $A_i = \supp(\cA(i, :))$. If the outcome is positive, $A_i$ must contain at least one item from each of $S_1$ and $S_2$. We then skip the current index and move to the next index. On the other hand, when the outcome is negative, there are three possibilities for $A_i$:
\begin{itemize}
\item $A_i \cap S_1 = \emptyset$ and $A_i \cap S_2 = \emptyset$.
\item $A_i \cap S_1 = \emptyset$ and $A_i \cap S_2 \neq \emptyset$.
\item $A_i \cap S_1 \neq \emptyset$ and $A_i \cap S_2 = \emptyset$.
\end{itemize} 
As we have outlined, being  in the second and third of these cases will be useful for applying standard group testing.  
To help distinguish between the three cases, we note that the equation $\bR^i_{\cA} = \bY_{\cB}$ (i.e., adding $A_i$ to the tests performed on $\cB$ does not impact any outcomes) occurs if and only if one of the following holds:
\begin{itemize}
\item $A_i \cap S_1 = \emptyset$ and $A_i \cap S_2 = \emptyset$.
\item For all $j = 1, \ldots, t_{\cB}$, we have $B_j \cap S_1 \neq \emptyset$ and $B_j \cap S_2 \neq \emptyset$.
\item $A_i \cap S_1 = \emptyset$, $A_i \cap S_2 \neq \emptyset$,  and for all $j = 1, \ldots, t_{\cB}$ we have $B_j \cap S_1 = \emptyset$ and $B_j \cap S_2 \neq \emptyset$.
\item $A_i \cap S_1 \neq \emptyset$, $A_i \cap S_2 = \emptyset$, and for all $j = 1, \ldots, t_{\cB}$ we have $B_j \cap S_1 \neq \emptyset$ and $B_j \cap S_2 = \emptyset$.
\end{itemize}
Therefore, when $\bR^i_{\cA} \neq \bY_{\cB}$, there are only two possibilities:
\begin{itemize}
\item $A_i \cap S_1 = \emptyset$, $A_i \cap S_2 \neq \emptyset$, and for some $j \in \{ 1, \ldots, t_{\cB} \}$ we have $B_j \cap S_1 \neq \emptyset$ and $B_j \cap S_2 = \emptyset$.
\item $A_i \cap S_1 \neq \emptyset$, $A_i \cap S_2 = \emptyset$, and for some $j \in \{ 1, \ldots, t_{\cB} \}$ we have $B_j \cap S_1 = \emptyset$ and $B_j \cap S_2 \neq \emptyset$.
\end{itemize}
In both of these cases, there exists $j \in [t_{\cB}]$ such that $\test(B_j) = 0$ and $\test(A_i, B_j) = 1$. For the first possibility, because of the construction of $\bF^i_{\cA}$ in Figure \ref{fig:Enc_Rand} combined with $A_i \cap S_1 = \emptyset$ and $A_i \cap S_2 \neq \emptyset$, we must have $\bF^i_{\cA} = \cM \odot \bX_2$ (cf.~\eqref{eqn:classicalConversion}). Therefore, $\mathrm{Decode}(\bF^i_{\cA}, \cM)$ returns $S_2 = \supp(\bX_2)$. Similarly, for the second possibility, $\mathrm{Decode}(\bF^i_{\cA}, \cM)$ returns $S_1 = \supp(\bX_1)$. In summary, if the procedure in Fig.~\ref{fig:Dec_Rand} reaches the step of finding $j \in t_{\cB}$ satisfying either of the above two dot points, then it produces $S_1$ and $S_2$ successfully (provided that both standard GT subroutines succeed).

By Lemma~\ref{lem:S1S2}, the probability that none of the $A_i \cap S_1 \cap S_2$ contain only items in $S_1$ for $i = 1, \ldots, t_{\cA}$ (called event $\cE_A$), or none of $B_i \cap S_1 \cap S_2$ the contain only items in $S_2$ for $i = 1, \ldots, t_{\cB}$ (called event $\cE_B$) is at most $\epsilon_1 + \epsilon_2$. Let $\cE_C$ be the event that \emph{either} of the two invocations of $\mathrm{Decode}(\bZ = \cM \odot \bX, \cM)$ fail to recover $\bX$ when $|\bX| \leq s_{\max}$. By a union bound over the two invocations of decoding, we have $\Pr[\cE_C] \leq 2\epsilon_0$, where $\epsilon_0$ was characterized in \eqref{eq:eps0} and its subsequent discussion. Therefore, the probability that $S_1$ and $S_2$ are not recovered is at most
\begin{align}
\Pr[\cE_A^c \cap \cE_B^c] \cdot \Pr[\cE_C] + \Pr[\cE_A \cup \cE_B] \leq (1 - \epsilon_1 - \epsilon_2) 2\epsilon_0 + \epsilon_1 + \epsilon_2 \leq 2\epsilon_0 + \epsilon_1 + \epsilon_2. \label{eqn:TotalProb}
\end{align}
Since we have defined $\epsilon/4 = \max\{\epsilon_0, \epsilon_1, \epsilon_2 \}$, we get $2\epsilon_0 + \epsilon_1 + \epsilon_2 \le \epsilon$. 

\textbf{Complexity:} When $\epsilon_1 = \epsilon_2 = \epsilon^\prime$ (which is at most $\frac{\epsilon}{4}$), the number of tests is
\begin{align}
&(t_{\cA} + t_{\cA} t_{\cB} + t_{\cA} m) + (t_{\cB} + t_{\cB} m) \nonumber \\
&= \left( \frac{\mathrm{e}^5}{4 \pi^2} \log{\frac{1}{\epsilon_1}} \cdot \frac{s}{s_1} + \frac{\mathrm{e}^5}{4 \pi^2} \log{\frac{1}{\epsilon_2}} \cdot \frac{s}{s_2} + 1 \right) \times O(s_{\max} \log{n}) + \frac{\mathrm{e}^5}{4 \pi^2} \log{\frac{1}{\epsilon_1}} \cdot \frac{s}{s_1} \times \frac{\mathrm{e}^5}{4 \pi^2} \log{\frac{1}{\epsilon_2}} \cdot \frac{s}{s_2} \\
&= O \left( \frac{s_{\max}^2}{s_{\min}} \log{n} \cdot \log{\frac{1}{\epsilon^\prime}} \right) + O \left( \frac{s_{\max}}{s_{\min}} \log^2{\frac{1}{\epsilon^\prime}} \right) \\
&= O \left( \frac{s_{\max}^2}{s_{\min}} \log{n} \right), 
\end{align}
where the last line holds because $\epsilon^\prime$ is a constant.

The decoding time has the same order as the number of tests because the time for running $\mathrm{Decode}(\bZ = \cM \odot \bX, \cM)$ with $|\bX| \leq s_{\max}$ is linear with respect to the number of rows in $\cM$, i.e., $m$~\cite{cheraghchi2020combinatorial}. Therefore, the decoding time of our proposed algorithm is $O\big( \frac{s_{\max}^2}{s_{\min}} \log{n} \big)$. This completes the proof of the first part of Theorem~\ref{thm:rnd}.

\subsection{Two-stage design}
\label{sec:rnd:2Stages}

Next we design a two-stage procedure to identify $S_1$ and $S_2$, illustrated in Fig.~\ref{fig:EncDec2}. 

\subsubsection{Encoding and decoding procedure}
\label{sec:rnd:2Stage:EncDec}

In the first stage, we generate two indicator matrices $\cA$ and $\cB$ with their corresponding indicator vectors $\bY_{\cA}$ and $\bY_{\cB}$, $t_{\cA} t_{\cB}$ tests, and the same matrix $\cM$ as in Section~\ref{sec:rnd} (i.e., one from \cite{cheraghchi2020combinatorial}), with $t_{\cA}$ and $t_{\cB}$ chosen using Lemma \ref{lem:S1S2}.  During decoding, we find indices $i \in [t_{\cA}]$ and $j \in [t_{\cB}]$ such that $\test(A_i) = 0$, $\test(B_j) = 0$, and $\test(A_i, B_j) = 1$. In the second stage, we get $S_1$ and $S_2$ by calling $\mathrm{Decode}(\test(\cM, A_i), \cM)$ and $\mathrm{Decode}(\test(\cM, B_j), \cM)$.

\begin{figure}
\centering
\includegraphics[scale=0.5]{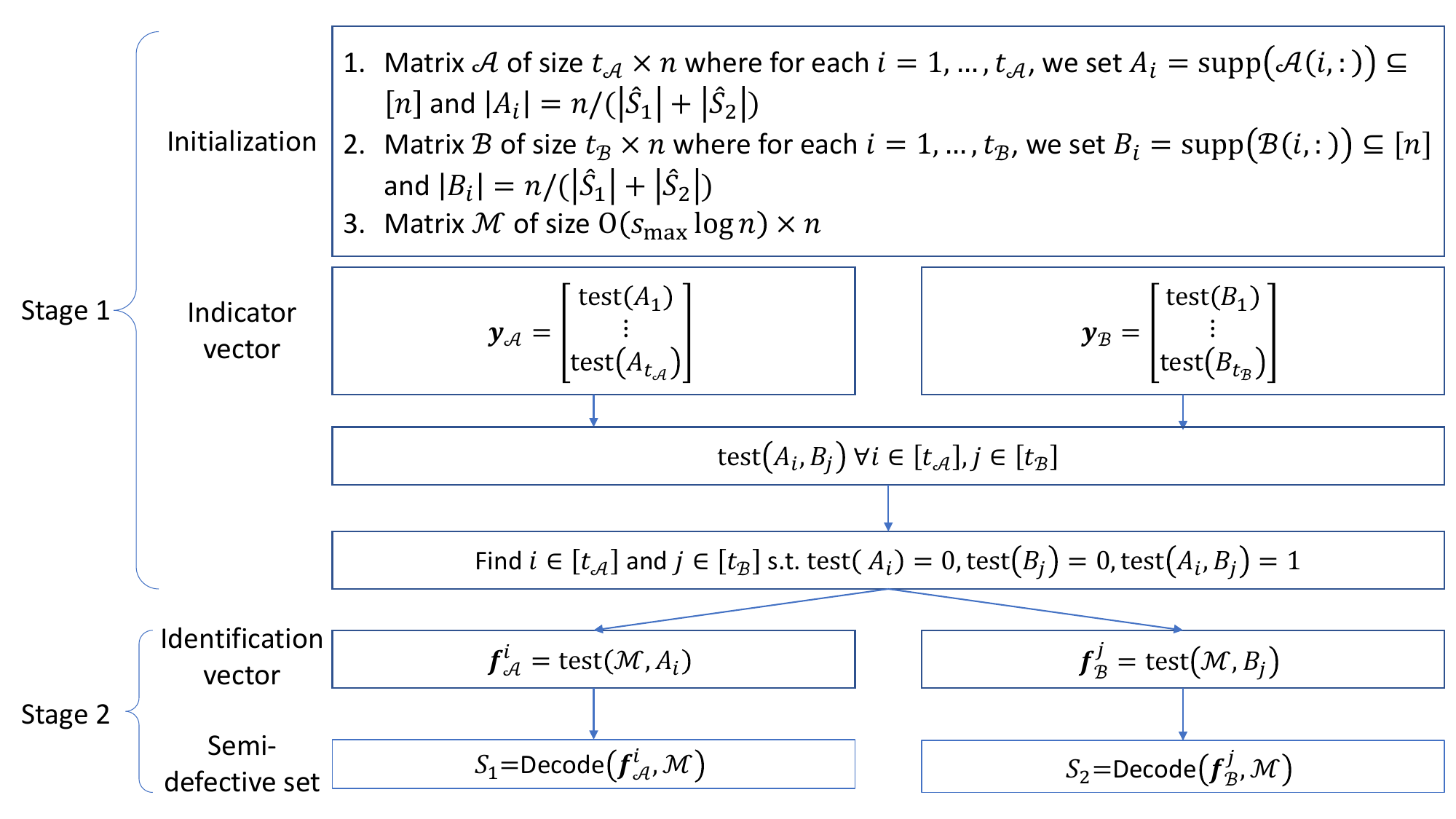}
\caption{Encoding and decoding procedure for the two-stage and randomized design.}
\label{fig:EncDec2}
\end{figure}

\subsubsection{Correctness}
\label{sec:rnd:2Stage:Correct}

In the first stage, if we can find $(i,j)$ such that $\test(A_i) = 0$, $\test(B_j) = 0$, and $\test(A_i, B_j) = 1$, then $A_i$ must contain some item in $S_1$ and none in $S_2$, and $B_j$ must contain some item in $S_2$ and none in $S_1$, or vice versa. Therefore, in the second stage, $S_1$ and $S_2$ are identified by calling $\mathrm{Decode}(\test(\cM, A_i), \cM)$ and $\mathrm{Decode}(\test(\cM, B_j), \cM)$. Similar to the analysis of our non-adaptive design, the overall error probability is at most $\epsilon$.

\subsubsection{Complexity}
\label{sec:rnd:2Stage:complx}

When $\epsilon_1 = \epsilon_2 = \epsilon^\prime$, the number of tests is
\begin{align}
&t_{\cA} + t_{\cB} + t_{\cA} t_{\cB} + 2m \\
&= \left( \frac{\mathrm{e}^5}{4 \pi^2} \log{\frac{1}{\epsilon_1}} \cdot \frac{s}{s_1} + \left( \frac{\mathrm{e}^5}{4 \pi^2} \log{\frac{1}{\epsilon_2}} \right) \cdot \left( \frac{s}{s_2} \right) + \frac{\mathrm{e}^5}{4 \pi^2} \log{\frac{1}{\epsilon_1}} \cdot \frac{s}{s_1} \times \frac{\mathrm{e}^5}{4 \pi^2} \log{\frac{1}{\epsilon_2}} \cdot \frac{s}{s_2} \right) + O(s_{\max} \log{n}) \\
&= O(s_{\max} \log{n}) + O \left( \frac{s_{\max}}{s_{\min}} \log^2{\frac{1}{\epsilon^\prime}} \right)\\
&= O(s_{\max} \log{n}), 
\end{align}
because $\epsilon^\prime$ is a constant. 
The decoding time again has the same order as the number of tests due to \cite{cheraghchi2020combinatorial}. Therefore, the decoding time of our proposed algorithm is $O \left( s_{\max} \log{n} \right)$.  We have thus established the second part of Theorem \ref{thm:rnd}.

\subsection{Three-stage design}
\label{sec:rnd:3Stage}

\subsubsection{Encoding and decoding procedure}
\label{sec:rnd:3Stage:EncDec}

In this section, we design a three-stage procedure to identify $S_1$ and $S_2$.  The first stage is identical to our two-stage design, but we now improve the standard group testing subroutine by having it span a further two stages.  Specifically, let $\mathrm{Decode}^\prime(\bZ = \cM \odot \bX, \cM)$ be a deterministic 2-stage procedure that returns vector $\bX$ where $|\bX| \leq s_{\max}$. It has been shown that there exists such an $m \times n$ matrix $\cM$ with $m = O \big( s_{\max} \log{\frac{n}{s_{\max}}} \big)$, which is information-theoretically optimal, that can be decoded in time $O \big( s_{\max} \log^2{\frac{n}{s_{\max}}} \big)$~\cite[Theorem 3.2]{cheraghchi2020combinatorial}.\footnote{We could alternatively use the design of de Bonis \cite{de2005optimal} (as we did in our fully adaptive design) and attain the same scaling in the number of tests but with more explicit constants.  However, we use \cite{cheraghchi2020combinatorial} here because it comes with the advantage of improved decoding time.}  We recover $S_1$ and $S_2$ by calling $\mathrm{Decode}^\prime(\test(\cM, A_i), \cM)$ and $\mathrm{Decode}^\prime(\test(\cM, B_j), \cM)$, and this constitutes our second and third stages with the two invocations being run in parallel.

\subsubsection{Correctness}
\label{sec:rnd:3Stage:Correct}

Similar to the proof in Section~\ref{sec:rnd:2Stage:Correct}, once $\test(A_i) = 0$, $\test(B_j) = 0$, $A_i$ must contain some item in $S_1$ and none in $S_2$, and $B_j$ must contain some item in $S_2$ and none in $S_1$, or vice versa. Since the second and third stages are deterministic, $S_1$ and $S_2$ are recovered.

The event that there does not exist $i \in [t_{\cA}]$ and $j \in t_{\cB}$ such that $\test(A_i) = 0$, $\test(B_j) = 0$, and $\test(A_i, B_j) = 1$ occurs if the events $\cE_A$ or $\cE_B$ in Section~\ref{sec:rnd:1Stage:DecCorrect} occur, constituting an error probability of up to $\epsilon_1 + \epsilon_2$. Since the second and third stages have zero error probability \cite{cheraghchi2020combinatorial}, the probability that $S_1$ and $S_2$ are not recovered is at most $\epsilon_1 + \epsilon_2$.

\subsubsection{Complexity}
\label{sec:rnd:3Stage:complx}

Similar to Section~\ref{sec:rnd:2Stage:complx}, when $\epsilon_1 = \epsilon_2 = \epsilon^\prime$, the number of tests is $t_{\cA} + t_{\cB} + t_{\cA} t_{\cB} + m = O\big(s_{\max} \log{\frac{n}{s_{\max}}}\big)$ because $\epsilon^\prime$ is a constant.  Moreover, once the test results have been attained, the decoding time is
\begin{align}
&t_{\cA} + t_{\cB} + t_{\cA} t_{\cB} + O(m \log{m}) \\
&= \left( \frac{\mathrm{e}^5}{4 \pi^2} \log{\frac{1}{\epsilon_1}} \cdot \frac{s}{s_1} + \frac{\mathrm{e}^5}{4 \pi^2} \log{\frac{1}{\epsilon_2}} \cdot \frac{s}{s_2} \right) + \frac{\mathrm{e}^5}{4 \pi^2} \log{\frac{1}{\epsilon_1}} \cdot \frac{s}{s_1} \times \frac{\mathrm{e}^5}{4 \pi^2} \log{\frac{1}{\epsilon_2}} \cdot \frac{s}{s_2} + O\left( s_{\max} \log^2{\frac{n}{s_{\max}}} \right) \\
&= O\left( s_{\max} \log^2{\frac{n}{s_{\max}}} \right) + O \left( \frac{s_{\max}}{s_{\min}} \log^2{\frac{1}{\epsilon^\prime}} \right). \\
&= O\left( s_{\max} \log^2{\frac{n}{s_{\max}}} \right),
\end{align}
because $\epsilon^\prime$ is a constant.  We have thus established the third part of Theorem \ref{thm:rnd}.

\subsection{Las Vegas design}
\label{sec:rnd:LasVegas}

Here we present a Las Vegas design that uses a random number of stages and tests and attains zero error probability. The idea is as follows:
\begin{enumerate}
\item Run the ``find $(A_i,B_j)$ intersecting one of $(S_1,S_2)$ each'' subroutine (i.e., the same as the first stage of our 2-stage and 3-stage designs) with success probability $1 - \alpha$ for some $\alpha > 0$.  Note that by the construction of this subroutine, whenever it succeeds, the algorithm has 100\% certainty that it did so, since there is no other way to attain  $\test(A_i) = 0$, $\test(B_j) = 0$, and $\test(A_i, B_j) = 1$.
\item Repeat the previous step as needed until the first time it succeeds.
\item Proceed with the second and third stages of the three-stage design in Section~\ref{sec:rnd:3Stage}.
\end{enumerate}

The last step above (i.e., the final two stages) has zero error probability and is already well-characterized, so it remains to study the first two steps. For the this part, the number of failures until the first success follows the geometric distribution with parameter $1 - \alpha$, which means it has expectation $1/(1 - \alpha) = 1+O(\alpha)$.  In particular, by choosing $\alpha$ sufficiently small, we find that it constitutes $1+\eta$ average stages for any pre-specified $\eta > 0$.

To understand $\alpha$ itself, we use the same reasoning as the first stage in Section~\ref{sec:rnd:2Stage:EncDec}. Because of the analysis in~\eqref{eqn:TotalProb}, the probability that one does not find a defective pair is at most $\alpha = \epsilon_1 + \epsilon_2$.  Since $\epsilon_1$ and $\epsilon_2$ can be arbitrarily small, so can $\alpha$.  The (expected) number of tests and runtime also follows in an identical manner to the three-stage design from Section~\ref{sec:rnd:3Stage}.  We have thus established the final part of Theorem \ref{thm:rnd}.

\section{Conclusion}
\label{sec:cls}

In this paper, we introduced a new model called concomitant group testing that embodies the phenomenon that either (i) multiple defective components lead to an overall failure, or (ii) items in various groups ``cooperate'' to produce a given outcome.  We have developed algorithms and upper bounds on the number of tests for a number of variations (depending on the success criterion, number of adaptive stages, etc.), with optimal or near-optimal scaling in many cases of interest.  Beyond closing the remaining gaps (e.g., for non-adaptive deterministic designs), perhaps the most immediate direction for future research is to better understand the case that $m > 2$, particularly under non-adaptive, 2-stage, and 3-stage designs.  It may also be of interest to seek more precise constant factors in the number of tests for each setting.

\bibliographystyle{ieeetr}
\balance
\bibliography{bibli}

\appendices

\section{Proof of Lemma~\ref{lem:S1S2}}
\label{appx:sec:lem}

The probability of picking $a < s$ items in $S$ with $|S| = s$ when we sample $na/s$ items without replacement from $n$ items is~\cite[Lemma 4]{chan2013stochastic}
\begin{equation}
H\left(n, \frac{na}{s}, s, a\right) = \frac{\binom{n - s}{\frac{na}{s} - a} \binom{s}{a} }{\binom{n}{\frac{na}{s}} } \geq \frac{4 \pi^2}{\mathrm{e}^5} \frac{1}{\sqrt{a}} \sqrt{\frac{s}{s - a}} \sqrt{\frac{n}{n - s}} \geq \frac{4 \pi^2}{\mathrm{e}^5} \frac{1}{\sqrt{a}} .
\end{equation}
Moreover, conditioned on such an event, the probability of those items all belonging to $S_1$ is $ \frac{\binom{s_1}{a} }{\binom{s}{a} }$ due to symmetry.  Hence, the overall probability of picking $a$ items from $S_1$ and none from $S_2 = S \setminus S_1$ is
\begin{equation}
\frac{\binom{s_1}{a} }{\binom{s}{a} } H \left(n, \frac{na}{s}, s, a\right) \geq \frac{\binom{s_1}{a} }{\binom{s}{a} } \frac{4 \pi^2}{\mathrm{e}^5} \frac{1}{\sqrt{a}}. 
\end{equation}

In accordance with the statement of Lemma \ref{appx:sec:lem}, we set $a = 1$.  Then, by sampling $n/s$ items from $n$ items with replacement $t_1$ times, the probability that none of the $t_1$ trials yield the above condition (i.e, having one item from $S_1$ and none from $S \setminus S_1$) is at most
\begin{align}
\left( 1 - \frac{\binom{s_1}{1} }{\binom{s}{1} } H \left(n, \frac{n}{s}, s, 1 \right) \right)^{t_1} \leq \left( 1 - \frac{s_1}{s} \cdot \frac{4 \pi^2}{\mathrm{e}^5} \right)^{t_1} \leq \epsilon_1, \label{eqn:t1}
\end{align}
where $t_1 = \frac{\mathrm{e}^5}{4 \pi^2} \log{\frac{1}{\epsilon_1}} \cdot \frac{s}{s_1}$ for any $\epsilon_1 > 0$.

Similarly, by sampling $n/s$ items from $n$ items with replacement $t_2$ times, the probability that none of the $t_2$ trials yield an analogous condition (i.e., having one item from $S_2$ and none from $S_1$) is at most
\begin{align}
\left( 1 - \frac{\binom{s_2}{1} }{\binom{s}{1} } H \left(n, \frac{n}{s}, s, 1 \right) \right)^{t_2} \leq \left( 1 - \frac{s_2}{s} \cdot \frac{4 \pi^2}{\mathrm{e}^5} \right)^{t_2} \leq \epsilon_2, \label{eqn:t2}
\end{align}
where $t_2 = \frac{\mathrm{e}^5}{4 \pi^2} \log{\frac{1}{\epsilon_2}} \cdot \frac{s}{s_2}$ for any $\epsilon_2 > 0$.

\section{Relaxing the Assumption of Known Semi-Defective Set Sizes} \label{app:KnownConst}

In our analysis of randomized designs, the requirement of known cardinalities $s_1$ and $s_2$ enters via the preceding proof of Lemma~\ref{lem:S1S2}, and in particular the fact that the number of items sampled without replacement is exactly $\frac{n}{s}$ with $s = s_1 + s_2$.  By making this choice, we are assuming that the test designer knows $s$.   Other parts of our analysis (e.g., invocations of standard group testing subroutines such as \cite{cheraghchi2020combinatorial}) only require \emph{upper bounds} on $s_1$ and $s_2$ to be known, which is a much milder requirement.

In this appendix, we outline how the same scaling laws can be attained when $s_1$ and $s_2$ are only known to within constant factors.  
More specifically, suppose it is known that $s_1 \le |S_1| \le c s_1$ for some known constant $c > 1$, and similarly for $|S_2|$.  (Note that by rescaling, we can form $s'_1$ (and $s'_2$) such that $\frac{1}{\sqrt{c}} s'_1 \le |S_1| \le \sqrt{c} s_1$, or such that $\frac{1}{c} s'_1 \le |S_1| \le s_1$.)  To simplify the outline, we consider a variation of Lemma \ref{lem:S1S2} that samples $\frac{cn}{s}$ items \emph{with} replacement; this is an equally valid test design as the one without replacement, but is slightly easier to generalize to the case that $c > 1$.

If we sample $\frac{n}{s}$ items with replacement (with $s = s_1 + s_2$), then the probability of selecting a single item from $S = S_1 \cup S_2$ is exactly $\frac{n}{s} \cdot \frac{s^*}{n} \cdot \big(1-\frac{s^*}{n}\big)^{\frac{n}{s}-1}$, where $s^* = |S_1| + |S_2|$.  Since $s \le s^* \le cs$ by assumption, the preceding expression is lower bounded by $\big(1-\frac{cs}{n}\big)^{\frac{n}{s}-1}$.  If $s = \Theta(n)$ then this trivially scales as $\Theta(1)$, whereas if $s = o(n)$ then this asymptotically approaches $e^{-c}$, which is again $\Theta(1)$.  

Then, by the same argument as Lemma \ref{lem:S1S2}, the probability of getting a single item from $S_1$ and none from $S_2$ is $\Theta\big( \frac{s_1}{s} \big)$, and similarly if the roles of $S_1$ and $S_2$ are reversed.  This gives us the desired counterpart to Lemma \ref{lem:S1S2} (albeit with a potentially poor constant factor $e^{-c}$), and the other parts of our analysis for randomized designs remains unchanged.

\section{Full comparison table}
\label{appx:sec:fullCmp}

In Table~\ref{tbl:FullCmp}, we replicate the results from Table~\ref{tbl:cmp} but also add the sufficient number of tests that would be obtained by specializing graph learning results to our setting (see Section \ref{sub:intro:hiddengraphs}).  We see that such specializations tend to be highly suboptimal; in fairness, this is because they solved a more general problem rather than specifically seeking to solve concomitant GT.  In short, their increased generality comes at the expense of being highly suboptimal for our problem.

\begin{table}[]
\centering
\begin{tabular}{|c|c|c|l|c|}
\hline
\begin{tabular}[c]{@{}c@{}}Design\\ type\end{tabular} & No. of stages & \multicolumn{1}{c|}{\begin{tabular}[c]{@{}c@{}}Semi-defective\\ sets\end{tabular}} & \multicolumn{1}{c|}{Scheme} & No. of tests \\
\hline
\multicolumn{1}{|l|}{Arbitrary} & \multicolumn{1}{c|}{Any} & \begin{tabular}[c]{@{}c@{}} $|P_i| \leq s_i$ for \\ $i = 1, \ldots, m$ \end{tabular} & \begin{tabular}[c]{@{}l@{}} Any (Lower bound) \\\textbf{Theorem~\ref{thm:lower}} \end{tabular} & \multicolumn{1}{l|}{\begin{tabular}[c]{@{}l@{}} $\Omega \big( \sum_{i = 1}^m s_i \log{\frac{n}{s_i}} \big)$ if $\sum_{v = 1}^{m} s_v \leq \frac{n}{2}$ \end{tabular}} \\ \hline
\multirow{5}{*}{\begin{tabular}[c]{@{}l@{}} \\ \\ Deterministic \end{tabular}} & \multirow{2}{*}{1} & \multirow{4}{*}{\begin{tabular}[c]{@{}l@{}} \\ $|P_1| \leq s_1$ \\ $|P_2| \leq s_2$ \end{tabular}} & \textbf{Theorem~\ref{thm:nonadaptive}} & $O \left( s_{\max}^3 \log{\frac{n}{s_{\max}}} \right)$ \\
\cline{4-5} & & & Chen et al.~\cite{chen2009nonadaptive} & $O \left( (s_1 s_2)^3 \log{\frac{n}{s_1 s_2}} \right)$ \\
\cline{2-2} \cline{4-5} & 2 & & \begin{tabular}[c]{@{}l@{}} Abasi and Nader~\cite{abasi2019learning} \end{tabular} & $O\big( (s_1 s_2)^2 \log{n} \big)$ \\
\cline{2-2} \cline{4-5} & $\log_2{n}$ & & \multirow{2}{*}{\begin{tabular}[c]{@{}l@{}} \\ \textbf{Theorem}~\ref{thm:adaptive:multi} \end{tabular}} & $O\left( s_1 \log{\frac{n}{s_1}} + s_2 \log{\frac{n}{s_2}} \right)$ \\
\cline{2-3} \cline{5-5} & $\frac{m \log{n}}{1 - \log{2}} + 3$ & \begin{tabular}[c]{@{}c@{}} $|P_i| \leq s_i$ for \\ $i = 1, \ldots, m$ \end{tabular} & & $O\left( m^2 \log{n} + \sum_{i = 1}^m  s_i \log{\frac{n}{s_i} } \right)$ \\
\hline
\multirow{8}{*}{\begin{tabular}[c]{@{}l@{}} \\ \\ \\ \\ Random \end{tabular}} & \multirow{2}{*}{1} & \begin{tabular}[c]{@{}l@{}} $|P_1| = s_1$ \\ $|P_2|= s_2$ \end{tabular} & \begin{tabular}[c]{@{}l@{}} \textbf{Theorem}~\ref{thm:rnd}\\ (Monte Carlo)\end{tabular} & $O \big( s_{\max}^2 \log{\frac{n}{s_{\max}}} \big)$ \\
\cline{3-5} & & \begin{tabular}[c]{@{}l@{}} $|P_1| \leq s_1$ \\ $|P_2| \leq s_2$ \end{tabular} & \begin{tabular}[c]{@{}l@{}}Abasi and Nader~\cite{abasi2019learning}\\ (Monte Carlo)\end{tabular} & $O\big( (s_1 s_2)^2 \log{n} \big)$  \\
\cline{2-5} & \multirow{3}{*}{2} & \begin{tabular}[c]{@{}l@{}} $|P_1| = s_1$ \\ $|P_2| = s_2$ \end{tabular} & \begin{tabular}[c]{@{}l@{}} \textbf{Theorem}~\ref{thm:rnd}\\ (Monte Carlo)\end{tabular} & $O( s_{\max} \log{n} )$ \\
\cline{3-5} & & \multirow{2}{*}{\begin{tabular}[c]{@{}l@{}} $|P_1| \leq s_1$ \\ $|P_2| \leq s_2$ \end{tabular}} & \begin{tabular}[c]{@{}l@{}}Abasi and Nader~\cite{abasi2019learning}\\ (Monte Carlo)\end{tabular} & $O( (s_1 s_2)^{4/3} \log{n})$ \\
\cline{4-5} & & & \begin{tabular}[c]{@{}l@{}}Abasi and Nader~\cite{abasi2019learning}\\ (Las Vegas)\end{tabular} & $O\big( (s_1 s_2)^2 \log{n} \big)$ \\
\cline{2-5} & \multicolumn{1}{c|}{3 (Monte Carlo); avg. $3 + \eta$ (Las Vegas)} & \begin{tabular}[c]{@{}l@{}} $|P_1| = s_1$ \\ $|P_2| = s_2$ \end{tabular} & \textbf{Theorem}~\ref{thm:rnd} & \multicolumn{1}{c|}{$O \big( s_{\max} \log{\frac{n}{s_{\max}}} \big)$} \\
\cline{2-5} & \multirow{2}{*}{3} & \multirow{2}{*}{\begin{tabular}[c]{@{}l@{}} $|P_1| \leq s_1$ \\ $|P_2| \leq s_2$ \end{tabular}} & \begin{tabular}[c]{@{}l@{}}Abasi and Nader~\cite{abasi2019learning}\\ (Monte Carlo)\end{tabular} & $O\big( s_1 s_2 \log{n} \big)$ \\
\cline{4-5} & & & \begin{tabular}[c]{@{}l@{}}Abasi and Nader~\cite{abasi2019learning}\\ (Las Vegas)\end{tabular} & $O\big( (s_1 s_2)^{4/3} \log{n} \big)$ \\
\hline
\end{tabular}

\caption{Comparison of proposed schemes with previous ones (which solve a more general graph learning problem that can be specialized to our setting; see Section \ref{sub:intro:hiddengraphs}). Given a population of $n$ items, there are $m \geq 2$ disjoint nonempty semi-defective subsets $S_1, \ldots, S_m \subset [n]$ with $|S_i| \leq s_i$ for $i = 1, 2, \ldots, m$. When considering randomized design, we set $m = 2$, $|S_1| = s_1$, and $|S_2| = s_2$ for our proposed schemes. In the bottom row, $\eta> 0$ is an arbitrarily small constant.}
\label{tbl:FullCmp}
\end{table}

\end{document}